\documentclass[journal]{IEEEtran}
\usepackage{amsmath,amsfonts,amssymb}

\usepackage{algorithm}
\usepackage{algpseudocode}

\usepackage{booktabs}
\usepackage{multirow}
\usepackage{tabularx}
\usepackage{diagbox}

\usepackage[nocompress]{cite}
\usepackage[pdfa]{hyperref}
\usepackage{hyperref}

\usepackage[caption=false,font=footnotesize]{subfig}
\usepackage{graphicx}

\usepackage{tikz, pgfplots}
\usetikzlibrary{shapes, arrows, arrows.meta, fit, positioning, calc, 3d}
\usepgfplotslibrary{fillbetween}
\pgfplotsset{compat=1.18} 
\tikzset{>=latex}

\newtheorem{theorem}{Theorem}
\newtheorem{lemma}{Lemma}
\newtheorem{remark}{Remark}
\newtheorem{definition}{Definition}

\allowdisplaybreaks

\begin{document}
% FOR AVOIDING DASH FOR RECURRENT AUTHORS
\bstctlcite{IEEEtran:no_dash_repeated_names}

\title{Innovation-Based Sampling for New Information}
\author{Jiping~Luo, Anthony~Ephremides,~\IEEEmembership{Life~Fellow,~IEEE}, and Nikolaos~Pappas,~\IEEEmembership{Senior~Member,~IEEE}
\thanks{This work was supported by ELLIIT, CUGS, and the European Union (6G-LEADER, 101192080, and MAGIC-6G, 101292933).}
\thanks{Jiping Luo and Nikolaos Pappas are with the Department of Computer and Information Science, Link\"{o}ping University, 58183 Link\"{o}ping, Sweden (e-mail: jiping.luo@liu.se; nikolaos.pappas@liu.se).}
\thanks{Anthony Ephremides is with the Electrical and Computer Engineering Department, University of Maryland, College Park, MD 20742 USA (e-mail: etony@umd.edu).}
}

\maketitle
\begin{abstract}
This work introduces innovation-based sampling of continuous stochastic processes, in which a sample is generated only when the process reveals ``sufficiently new'' information relative to all previous samples. The resulting innovation process admits a lattice structure and a three-dimensional state representation. For a Wiener process, we characterize the direction, timing, and frequency of innovations. We show that direction reversals become increasingly rare and establish limit theorems for their number. Innovations also become progressively sparser: the expected number of innovations grows only as the square root of time, and hence the sampling rate vanishes asymptotically. We then study remote estimation from sparsely received innovations. Notably, the absence of an innovation carries information: the minimum mean-square error (MMSE) estimate evolves with the age of information (AoI) and achieves a substantial MSE reduction over the conventional silence-ignorant zero-order hold (ZOH) estimator. Finally, we develop tractable affine-age and exponential-age approximations for practical use. Overall, information is conveyed not only by the content of innovations, but also by their direction and timing.
\end{abstract}

\section{Introduction}
Future communication systems are increasingly expected to deliver only timely and task-relevant information~\cite{howard1966information, vitturi2013industrial, gielis2022critical, kountouris2021semantics}. In emerging cyber-physical systems, communicating every piece of raw data is wasteful, if not infeasible. This challenge has motivated growing interest in semantics-aware and goal-oriented communications, where a central question concerns how much sampling and communication are necessary and what information is worth conveying~\cite{luo2025information}. For example, in stock trading or environmental monitoring, new record values, such as unprecedented price movements, temperatures, atmospheric pressures, or river levels, may be more informative than repeated observations within an already explored range.

In this work, we approach this question through the notion of \emph{innovation}, a piece of ``sufficiently new'' information that would surprise the receiver because it has not been observed before or cannot be predicted from previous samples. Intuitively, information may lose its significance through repetition and predictability, whereas excursions into previously unexplored regions reveal new information. As we show in this paper, an innovation conveys information not only through its content but also through its direction and timing.

This work studies innovation-based sampling of continuous stochastic processes. Its purpose is to communicate new information while suppressing repeated excursions within the explored range. We characterize the occurrence times and frequency of innovations, and their directions and reversals. We further study reconstruction between successive innovations to quantify what the receiver can infer from both the communicated samples and periods of silence.

\begin{figure}[t!]
    \centering
    \includegraphics[width=0.95\linewidth]{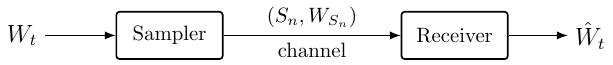}
    \caption{System model.}
    \label{fig:system_model}
\end{figure}

\subsection{Main Contributions}
For clarity, we briefly describe the system in Fig.~\ref{fig:system_model}. A sampler observes a continuous-time process $(W_{t})_{t\geq 0}$ and generates samples at times $(S_{n})_{n\geq 1}$. At each time $t$, a remote receiver produces an estimate $\hat{W}_{t}$ using the samples received by time $t$. Under innovation-based sampling, a sample is generated whenever the process exceeds the recorded maximum by a prescribed threshold $\delta$, i.e., an upper innovation, or falls below the recorded minimum by $\delta$, i.e., a lower innovation. The resulting sample sequence $(W_{S_{n}})_{n\geq 1}$ is called the innovation process. Let $T_{n}=S_{n+1}-S_{n}$ denote the $n$th inter-sampling time, $N_{t}$ the number of innovations generated by time $t$, and $R_{t}=\mathbb{E}[N_{t}]/t$ the expected sampling rate. Let $D_{n}$, $I_{n}$, and $V_{n}$ denote, respectively, the direction of the $n$th innovation, its direction-reversal indicator, and the number of reversals among the first $n$ innovations.

Our main contributions are as follows.

\subsubsection{Innovation-Based Sampling}
We introduce innovation-based sampling, which generates a sample only when the source enters a sufficiently new region relative to all previous samples. The resulting innovation process lies on a lattice, and its spread grows linearly with the number of innovations. Despite its history dependence, the sampler admits a three-dimensional state representation and therefore need not retain the full sampling history.

We then specialize the framework to a Wiener process.

\subsubsection{Characterization of the Innovation Process}
We show that the directions $(D_{n})_{n\geq 1}$ form a time-inhomogeneous binary Markov chain. The probability of a direction reversal at the $n$th innovation is $\mathbf{P}(I_{n}=1)=\frac{1}{n+1}$, yielding a reversal surprisal of $\log_{2}(n+1)$. We further show that the reversal indicators $(I_{n})_{n\geq 2}$ are mutually independent and establish a strong law of large numbers and a central limit theorem for the number of reversals $V_{n}$.

We also derive the distributions and moments of the sampling times $(S_{n})_{n\geq 1}$ and inter-sampling times $(T_{n})_{n\geq 1}$. The expected inter-sampling time $\mathbb{E}[T_{n}]$ grows linearly with $n$, whereas the expected sampling time $\mathbb{E}[S_{n}]$ grows quadratically. Moreover, the expected number of innovations $\mathbb{E}[N_{t}]$ grows on the order of $\sqrt{t}$, while the sampling rate $R_{t}$ decays on the order of $t^{-1/2}$ and vanishes asymptotically. This behavior is in sharp contrast to classical periodic and deviation-based rules, which have constant asymptotic sampling rates.

\subsubsection{Age-Aware Estimation}
We then study remote estimation under innovation-based sampling and derive the minimum mean-square error (MMSE) estimator. Notably, we show that silence is informative and that the MMSE estimate evolves with the age of information (AoI), i.e., the elapsed time since the latest received sample. In contrast, under deviation-based sampling, the MMSE estimator reduces to the silence-ignorant zero-order hold (ZOH) estimator, which retains the latest received sample until the next update. By exploiting silence, the MMSE estimator achieves an asymptotic MSE reduction of $82.14\%$ relative to ZOH.

Because the MMSE estimator involves an infinite series, we develop analytically tractable approximations. We introduce affine-age estimators whose correction terms vary linearly with AoI and exponential-age estimators whose corrections converge to the MMSE long-silence limit. For both classes, we derive closed-form MSE expressions and obtain the optimal coefficients. Relative to ZOH, the optimal affine-age and exponential-age estimators achieve asymptotic MSE reductions of $74.29\%$ and $81.85\%$, respectively.

\subsection{Article Organization}
Section~\ref{sec:related-work} reviews related work on goal-oriented semantic communications, sampling rules, and record theory. Section~\ref{sec:system-model} presents the system model and formalizes innovation-based sampling. Section~\ref{sec:sampling-wiener-process} characterizes the innovation directions, reversals, sampling times, and sampling rate for a Wiener process. Section~\ref{sec:estimation-wiener-process} develops the MMSE and low-complexity silence-aware estimators and analyzes their estimation performance. The Appendix collects the necessary background on Wiener processes and stochastic integrals used in our analysis.

\section{Related Work}\label{sec:related-work}
\subsection{Goal-Oriented Semantic Communications}
A central purpose of goal-oriented semantic communication is to quantify the value of information, so as to selectively sample and transmit only what contributes to the system's goal. The notion of AoI is the first concrete step in this direction~\cite{kaul2012real}. It holds the premise that information is valuable when it is fresh. This perspective has reshaped communication system design, including sampling, queueing, transmission scheduling, and random access protocols~\cite{yates2021age, sun2022age, pappas2023age}. However, freshness alone does not determine usefulness: a newly generated sample may repeat known information or have little value to the end user. 

This limitation has motivated semantics-aware metrics that capture information attributes such as content significance, contextual relevance, and error persistence. For example, information versions represent significant changes in content, while version age measures the number of versions by which the receiver lags behind the source~\cite{yates2021Vage, buyukates2022version, delfani2025version, kaswan2025age, salimnejad2025age, tekez2026information}. Error-persistence metrics assign greater importance to updates that correct prolonged errors, which may otherwise incur severe operational risks or costs~\cite{maatouk2020age, salimnejad2023state, chen2024minimizing, cosandal2025multi, luo2024exploiting, bastopcu2022using, zakeri2024semantic, chiariotti2025distributed}. Context-aware measures further distinguish updates according to their consequences, utility, and urgency~\cite{zheng2020urgency, pappas2021goal, ornee2023context, luo2025semantic, wang2026information, luo2025cost, luo2026role}. This line of research is surveyed in~\cite{luo2025information}. Complementing these metrics, the Pareto-optimal design quantifies the marginal value of communication and identifies the minimum communication needed to achieve the goal~\cite{luo2026value}.

However, much of the literature presumes an update generation model, such as an independent and identically distributed (i.i.d.) process or a time-homogeneous Markov chain~\cite{luo2025information}. In contrast, this work introduces innovation\footnote{In filtering theory, an innovation conventionally denotes the measurement prediction error~\cite{kailath1968innovations}. We extend this notion beyond unpredictability to novelty.} as a criterion for generating informative updates directly from raw observations. For a Wiener source, we show that the resulting update process is nonstationary and becomes progressively sparser.

\subsection{Sampling Rules}
Periodic sampling is the classical approach in digital communication and control, in which measurements are acquired at uniformly spaced time instants. Its theoretical foundation is rooted in the Nyquist--Shannon theorem~\cite{shannon1949communication}, which guarantees perfect reconstruction of bandlimited signals sampled above the Nyquist rate. Many stochastic processes of interest, including the Wiener process considered here, do not satisfy the bandlimited assumption. More importantly, periodic sampling does not adapt to the realized system trajectory and may therefore generate unnecessary samples when the system requires little attention~\cite{heemels2012introduction}.

This limitation has motivated various forms of signal-aware aperiodic sampling. Event-triggered sampling generates a sample when a prescribed state- or error-dependent condition is met~\cite{heemels2012introduction}. That is, sampling is triggered in response to system events rather than by elapsed time, thereby reducing sampling effort while meeting the system's goal. Related approaches include state-dependent sampling~\cite{fiter2012state}, which adapts the inter-sampling time to the sampled state while preserving control stability, and quantization-based sampling~\cite{bini2014optimal}, which nonuniformly allocates sampling instants to approximate the optimal continuous-time control input. 

For stochastic processes, a prominent signal-aware mechanism is deviation-based, or Lebesgue, sampling, which generates a sample when the process deviates from its most recently sampled value by a prescribed threshold. Åström and Bernhardsson showed that Lebesgue sampling can outperform periodic sampling at the same average sampling rate~\cite{astrom2002comparison}. Subsequent studies formulated causal sampling of the Wiener process as an optimal stopping problem and established the optimality of deviation-based sampling under ZOH estimation and a sampling rate constraint~\cite{rabi2012adaptive, nar2014sampling, sun2020sampling, guo2022optimal, tang2022sampling, pan2023sampling}.

\begin{figure}[t!]
    \centering
    \subfloat[periodic sampling]{
    \includegraphics[width=0.95\linewidth]{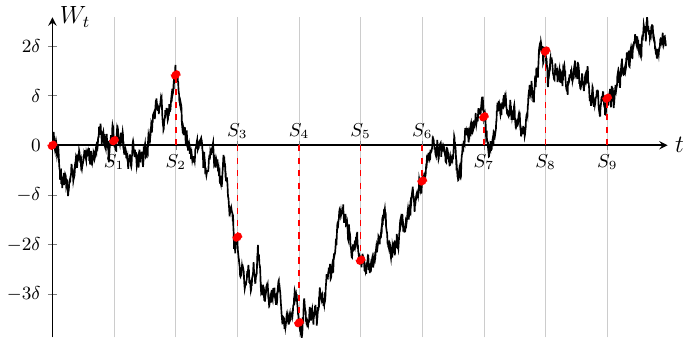}
    \label{fig:periodic_sampling}
    }
    \\
    \subfloat[deviation-based sampling]{
    \includegraphics[width=0.95\linewidth]{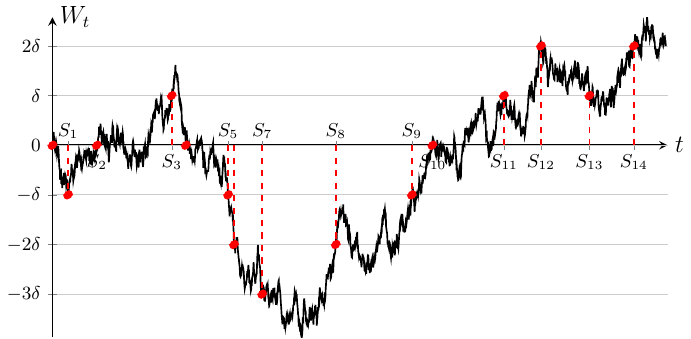}
    \label{fig:threshold_sampling}
    }
    \\
    \subfloat[innovation-based sampling]{
    \includegraphics[width=0.95\linewidth]{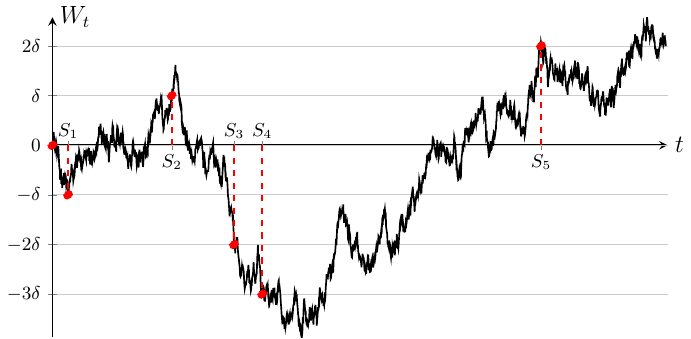}
    \label{fig:innovation_sampling}
    }
    \\
    \caption{Illustration of different sampling rules.}
    \label{fig:sampling_rules}
\end{figure}

Innovation-based sampling differs in the reference used to define a sampling event, as illustrated in Fig.~\ref{fig:sampling_rules}. Deviation-based sampling is local: after each sample, the triggering mechanism resets and compares the current process value with the most recently sampled value. Innovation-based sampling is history-dependent: a sample is generated when the process enters a previously unexplored region relative to the entire sampling history. Hence, the innovation sample sequence is a pathwise thinning of the deviation-based sample sequence. Crucially, silence is informative under innovation-based sampling, and the MMSE estimate evolves with AoI.

\subsection{Record Theory}
Our notion of innovation is closely related to record theory. Classical record theory, initiated by Chandler~\cite{chandler1952distribution}, studies observations that exceed or fall below all preceding observations, together with their values, occurrence times, and frequencies. For an i.i.d. sequence with a continuous distribution, let $\xi_{n}$ indicate that the $n$th observation is an upper record. R\'enyi~\cite{renyi1962theorie} proved that the record indicators $(\xi_{n})_{n\geq 1}$ are mutually independent with $\mathbf{P}(\xi_{n}=1)=\frac{1}{n}$. Consequently, the number of records among the first $n$ observations has mean and variance of order $\log(n)$, satisfies a strong law of large numbers, and is asymptotically Gaussian. Extensions to discrete distributions, independent but non-identically distributed sequences, and correlated sequences generated by random walks are surveyed in~\cite{arnold2011records, ahsanullah2015records}.

Our framework differs from classical record theory in two fundamental respects. First, innovations are generated endogenously from a continuous-time stochastic process, rather than identified from an exogenously given sequence of observations. Second, innovation-based sampling is two-sided: a sample is generated when the process expands either the previously recorded maximum or minimum by a prescribed threshold. Characterizing the innovation process therefore requires a first-exit time analysis. Interestingly, for a Wiener process, we show that the reversal indicators share the same probability structure as the record indicators. This connection allows record-theoretic results to be transferred directly to the number and occurrence indices of direction reversals.

\begin{table}[t!]
\caption{Notation Summary}
\label{table:notation}
\centering
\renewcommand{\arraystretch}{1.25}
\setlength{\tabcolsep}{3pt}
\begin{tabular}{|c|l|}
\hline
\textbf{Symbol}
& \multicolumn{1}{c|}{\textbf{Description}} \\
\hline
\rule{0pt}{2.6ex}
$W_t$, $\hat{W}_t$, $\widetilde{W}_t$
& Source process, its estimate, and the shifted process\\
\hline
$S_n$, $T_n$
& Sampling time and inter-sampling time \\
\hline
$R_t$, $R_\infty$
& Expected and asymptotic sampling rates \\
\hline
$U_t$
& Receiver's AoI, i.e., elapsed time since last update \\
\hline
\rule{0pt}{2.5ex}
$\overline{M}_t$, $\underline{m}_t$
& Running maximum and minimum of the process \\
\hline
$M_n$, $m_n$
& Recorded maximum and minimum of the samples \\
\hline
\multirow{2}{*}{$D_n$, $\Xi_n$, $N_t$}
& Innovation direction, no. of upper innovations, \\[-0.2em] & and no. of innovations generated by time $t$\\
\hline
\multirow{2}{*}{$I_n$, $V_n$, $\Upsilon_n$} & Indicator of a direction reversal, no. of reversals, \\[-0.2em] & and reversal surprisal \\
\hline
$\alpha_n$, $\beta_n$
& Lower and upper exit distances after the $n$th sample \\
\hline
\rule{0pt}{2.6ex}
$\mathcal{J}_n$, $\bar{\mathcal{J}}_n$
& MSE and average MSE over the first $n$ intervals \\
\hline
\end{tabular}
\end{table}

\section{Innovation-Based Sampling Framework}\label{sec:system-model}
\subsection{System Model}
Consider the system in Fig.~\ref{fig:system_model}. The variables to be defined are summarized in Table~\ref{table:notation}. Let $(W_t)_{t\geq 0}$ be a continuous stochastic process starting at $W_{0} = 0$. A sampler observes $W$ and generates samples at an increasing sequence of times $0 = S_{0} < S_{1} < S_{2} < \cdots$. Let
\begin{equation}
    T_{n} := S_{n+1} - S_{n}
\end{equation}
denote the $n$th inter-sampling time (or sampling interval). 

We consider \emph{causal} samplers so that the decision to sample at time $t$ depends only on information available up to time $t$. Formally, define
\begin{equation}
    \mathcal{F}_{t} = \bigcap_{s>t} 
    \mathcal{F}^{0}_{s}, \quad \mathcal{F}^{0}_{t} := \sigma(W_{s}: 0 \leq s \leq t),
\end{equation}
where $\mathcal{F}^{0}_{t}$ is the $\sigma$-field of events that can be defined in terms of the process up to time $t$. Then, $(\mathcal{F}_{t})_{t\geq 0}$ is a right-continuous filtration, and causality requires each $S_{n}$ to be a \emph{stopping time} with respect to this filtration, i.e.,
\begin{equation}
\{S_{n} \leq t\} \in \mathcal{F}_{t}, \quad t \geq 0.
\end{equation}
Further details on filtrations and stopping times are provided in the Appendix.

Samples and their timestamps are sent to a remote receiver through a communication link, which is assumed to be instantaneous and error-free. Suppose that $S_{n} \leq t < S_{n+1}$. Let
\begin{equation}
    \mathcal{H}_{n} := \sigma ( S_{k}, W_{S_{k}}: 0 \leq k \leq n )
\end{equation}
denote the received history after the $n$th sample. The receiver reconstructs the source from the sequence of communicated samples using the MMSE estimation rule. Define the AoI at the receiver at time $t$ by
\begin{equation}
    U_{t} := t - S_{N_{t}},
\end{equation}
where
\begin{equation}
    N_{t} := \max\{n : S_{n} \leq t\} \label{eq:counting-process}
\end{equation}
is the number of innovations generated by time $t$. The absence of a new sample since $S_{n}$ can be expressed by
\begin{equation}
    \{t < S_{n+1}\} = \{T_{n} > U_{t}\}.
\end{equation}

At each time $t$, the receiver knows not only the values and timestamps of all samples received up to time $S_n$, but also that no new sample has been generated since then. Under event-triggered sampling, \emph{silence itself can convey information about the process state.} Specifically, an absence of a new sample implies that the sampling event has not occurred, thus restricting the possible evolution of the process since the last sample. Accordingly, the MMSE estimator is 
\begin{equation}
    \hat{W}_{t} := \mathbb{E}[W_{t} {\,}|{\,} \mathcal{H}_{n}, T_{n} > U_{t}], \quad t \in [S_{n}, S_{n+1}). \label{eq:mmse-estimator}
\end{equation}

\begin{definition}
An estimator is said to be \emph{silence-aware} if it exploits the information conveyed by the absence of a new sample. A \emph{silence-ignorant} estimator is defined as
\begin{equation}
    \hat{W}_{t} := \mathbb{E}[W_{t} {\,}|{\,} \mathcal{H}_{n}], \quad t \in [S_{n}, S_{n+1}). \label{eq:mmse-estimator-si}
\end{equation}
\end{definition}

The accumulated MSE over the first $n$ sampling intervals is defined as
\begin{equation}
    \mathcal{J}_n
    :=
    \sum_{k=0}^{n-1}\mathsf{MSE}_k,
    \,\,
    \mathsf{MSE}_k
    :=
    \mathbb{E}\left[
        \int_{S_k}^{S_{k+1}}
        (W_t-\hat{W}_t)^2\,dt
    \right].
\end{equation}
The (ratio-)average MSE over the first $n$ sampling intervals is
\begin{equation}
    \bar{\mathcal{J}}_n
    :=
    \frac{\mathcal{J}_n}{\mathbb{E}[S_n]}
    =
    \frac{
        \sum_{k=0}^{n-1}\mathsf{MSE}_k
    }{
        \sum_{k=0}^{n-1}\mathbb{E}[T_k]
    }, \quad n\geq 1.
\end{equation}

We next review two classical sampling rules and then introduce the innovation-based sampling rule. These rules are illustrated on a common sample path in Fig.~\ref{fig:sampling_rules}.

\subsection{Two Classical Sampling Rules}
We first review two classical sampling rules.
\begin{itemize}
    \item \emph{Periodic rules} generate samples at equally spaced times, such that for some $\Delta_{\mathrm{S}} > 0$,
    \begin{equation}
        S_{n+1} = S_{n} + \Delta_{\mathrm{S}}.
    \end{equation}
    \item \emph{Deviation-based rules}~\cite{sun2020sampling, guo2022optimal} generate samples whenever the process deviates from the last sampled value by some fixed threshold $\delta > 0$, i.e., 
    \begin{equation}
        S_{n+1} = \inf \left\{t > S_{n}: |W_{t} - W_{S_{n}}| \geq \delta \right\}.
    \end{equation}
\end{itemize}

Both rules ignore full sampling history: periodic sampling ignores the process trajectory, while deviation-based sampling retains only the most recent sampled value. In contrast, the innovation-based sampling rule introduced next incorporates information about the range explored by all previous samples.

\subsection{Innovation-Based Sampling Rules}
Conceptually, an innovation represents ``sufficiently new'' information that would surprise the receiver. An innovation-based sampler suppresses measurements that remain within the range of previously communicated samples, thereby reducing the amount of less informative data generated and communicated over the network. We formalize this idea below.

\begin{definition}\label{def:extrema}
Define the \emph{running extrema} of the process by
\begin{equation}
    \overline{M}_{t} := \max_{0 \leq s \leq t} W_{s}, \qquad 
    \underline{m}_{t} := \min_{0 \leq s \leq t} W_{s}.
\end{equation}
After the $n$th sample is taken, define the \emph{recorded extrema} of the sampled values by
\begin{equation}
    M_{n} := \max_{0 \leq k \leq n} W_{S_{k}}, \qquad
    m_{n} := \min_{0 \leq k \leq n} W_{S_{k}},
\end{equation}
where $M_{0} = m_{0} = W_{0}$. We call $M_{n} - m_{n}$ the \emph{spread} of the sampled values.
\end{definition}

Using the recorded extrema as references, we define an innovation as an excursion beyond the previously sampled range by at least a prescribed threshold.

\begin{definition}
Fix a threshold $\delta>0$. We say that $W$ attains an \emph{innovation} after the $n$th sample if it exceeds the recorded maximum or falls below the recorded minimum by at least $\delta$, i.e., $W_{t} \geq M_{n} + \delta$ or $W_{t} \leq m_{n} - \delta$. An \emph{innovation-based sampling rule} generates samples only at innovation times; that is, $S_{0} = 0$ and for $n = 0, 1, \ldots$
\begin{equation}
    S_{n+1} = \inf\left\{t > S_{n}: W_{t} \notin \left(m_{n} - \delta, M_{n} + \delta\right)\right\}. \label{eq:sampling rule}
\end{equation}
The sequence $(W_{S_{n}})_{n \geq 1}$ is called the \emph{innovation process}. 
\end{definition}

To characterize how innovations evolve over time, we introduce quantities describing their frequency, direction, and direction reversals. We begin with the sampling rate.

\begin{definition}
Recall the counting process $(N_{t})_{t\geq 0}$ defined in~\eqref{eq:counting-process}. For $t>0$, let $R_t := \mathbb{E}[N_t] / t$ denote the \emph{expected sampling rate}. The \emph{asymptotic} sampling rate is $R_{\infty} = \lim_{t \to \infty} R_t$.
\end{definition}

We then define the direction of innovation. The $n$th sample is called an \emph{upper innovation} if it is a new recorded maximum, i.e., $W_{S_{n}} = M_{n-1} + \delta$, and a \emph{lower innovation} if it attains a new recorded minimum, i.e., $W_{S_{n}} = m_{n-1} - \delta$.

\begin{definition}
Let $D_{n} \in \{-1, +1\}$, $n \geq 1$, denote the direction of the $n$th innovation, where
\begin{equation}
    D_{n} := \begin{cases}
        +1, &W_{S_{n}}~\textrm{is an upper innovation},\\
        -1, &W_{S_{n}}~\textrm{is a lower innovation}.
    \end{cases}
\end{equation}
The number of upper innovations among the first $n$ innovations is given by
\begin{equation}
    \Xi_{n} := \sum_{k=1}^{n} \mathbb{I}\{D_{k} = +1\}.
\end{equation}
\end{definition}

Finally, we identify changes in innovation direction.

\begin{definition}
For $n\geq 2$, let $I_{n} := \mathbb{I}\{D_{n} \neq D_{n-1}\}$ indicate whether a \emph{direction reversal} occurs at the $n$th innovation. Let
\begin{equation}
    V_{n} := \sum_{k=2}^{n} I_{k}
\end{equation}
denote the number of reversals among the first $n$ innovations.
\end{definition}

\begin{remark}
Reversals carry additional semantic information beyond the occurrence of an innovation. A reversal indicates a change in the direction of exploration and may therefore warrant additional system inspection or intervention. 
\end{remark}

The information conveyed by a reversal is quantified below.

\begin{definition}
The \emph{reversal surprisal} is defined as
\begin{equation}
    \Upsilon_n := - \log_{2} \mathbf{P}(I_{n} = 1).
\end{equation}
\end{definition}

\subsection{Lattice Structure of the Innovation Process}
The next lemma shows a key structural property of the innovation-based sampler: the innovations lie on a lattice, and their spread grows linearly with the number of innovations. This structure serves as the basis for our analysis.

\begin{lemma}\label{lemma:spread}
The innovation process lies on the lattice $\delta \mathbb{Z}$, and its spread satisfies 
\begin{equation}
    M_{n} - m_{n} = n \delta, \quad n = 0, 1, 2 \ldots \label{eq:spread}
\end{equation}
\end{lemma}
\begin{IEEEproof}
Since $W_t$ is continuous, at each $S_{n}$ the process must hit one of the two boundaries exactly, i.e.,
\begin{equation*}
    W_{S_{n}} = \begin{cases}
        M_{n-1} + \delta, &\textrm{if an upper innovation occurs},\\
        m_{n-1} - \delta, &\textrm{if a lower innovation occurs}.
    \end{cases}
\end{equation*}
Therefore, every new innovation is either a new maximum or a new minimum of the sampled process. Since $M_{0}=m_{0}=0$, it follows by induction that $W_{S_{n}} \in \delta\mathbb{Z}$ for all $n$. 

We next derive the spread recursion. For an upper innovation, the recorded extrema satisfy $M_{n}=M_{n-1}+\delta$ and $m_{n}=m_{n-1}$, whereas for a lower innovation, $M_{n}=M_{n-1}$ and $m_{n}=m_{n-1}-\delta$. Therefore, in either case, 
\begin{equation*}
    M_{n}-m_{n} = M_{n-1} - m_{n-1} + \delta.
\end{equation*}
Since $M_{0} = m_0 =0$, induction gives~\eqref{eq:spread}.
\end{IEEEproof}

The lattice structure in Lemma~\ref{lemma:spread} implies that the innovation process admits a finite-dimensional representation. We formalize this result below.

\begin{lemma}
The innovation process evolves as
\begin{equation}
    W_{S_{n}} = \begin{cases}
    \Xi_{n} \delta, & D_{n}=+1,\\
    -(n - \Xi_{n}) \delta, & D_{n}=-1.
\end{cases}
\end{equation}
\end{lemma}

\begin{remark}
Crucially, history dependence does not entail growing memory complexity: instead of retaining the entire history $(W_{S_{0}}, W_{S_{1}}, \ldots, W_{S_{n}})$, it suffices to track the three-dimensional state $(n,\Xi_{n},D_{n})$. This property is useful in systems with stringent memory and computation constraints.
\end{remark}

% \begin{remark}
% Each innovation contains three types of information: value, timing, and reversal.
% \end{remark}

\begin{table}[t!]
\centering
\caption{Comparison of Deviation-Based and Innovation-Based Sampling of a Wiener Process}
\label{tab:sampling_comparison}
\renewcommand{\arraystretch}{1.2}
\setlength{\tabcolsep}{3pt}
\begin{tabular}{lcc}
\hline
\textbf{Property}
& \textbf{Deviation-based}
& \textbf{Innovation-based} \\
\hline

Reversal probability $\mathbf{P}(I_{n} = 1)$
& $1/2$
& $1/(n+1)$
\\

Reversal surprisal $\Upsilon_n$
& $1$
& $\log_{2}(n+1)$
\\

Expected no. of reversals $\mathbb{E}[V_{n}]$
& $(n-1)/2$
& $\log(n) +\mathcal{O}(1)$
\\

Expected sampling interval $\mathbb{E}[T_n]$
& $\delta^2$
& $(n+1)\delta^2$
\\

Expected sampling time $\mathbb{E}[S_n]$
& $n\delta^2$
& $n(n+1)\delta^2/2$
\\

Expected no. of samples $\mathbb{E}[N_t]$
& $t/\delta^2 + \mathcal{O}(1)$
& $\frac{1}{\delta}\sqrt{\frac{8t}{\pi}}
  + \mathcal{O}(1)$
\\

Expected sampling rate $R_t$
& $1/ \delta^2 + \mathcal{O}(t^{-1})$
& $\frac{1}{\delta}\sqrt{\frac{8}{\pi t}} + \mathcal{O}(t^{-1})$
\\

Asymptotic sampling rate $R_{\infty}$
& $1/ \delta^2$
& $0$
\\

\hline
\end{tabular}
\end{table}

\section{Main Results on Innovation-Based Sampling of a Wiener Process}\label{sec:sampling-wiener-process}
This section considers innovation-based sampling of a Wiener process. The main results are summarized in Table~\ref{tab:sampling_comparison}. Relevant preliminaries on Wiener processes and first exit times are provided in the Appendix.

\subsection{First-Exit Time Analysis}
Let $(W_{t})_{t \geq 0}$ be a standard Wiener process with respect to the filtration $(\mathcal{F}_t)_{t\geq 0}$, so that $W_{t} \sim \mathcal{N}(0, t)$ and
\begin{equation}
    W_{t} - W_{s} \sim \mathcal{N}(0, t - s), \quad 0 \leq s < t.
\end{equation}

By the strong Markov property (see Appendix~\ref{app:wiener-preliminaries}), for each stopping time $S_{n}$, the shifted process 
\begin{equation}
    \widetilde{W}_t := W_{S_{n} + t} - W_{S_{n}} \label{eq:shifted-process}
\end{equation}
is a standard Wiener process independent of $\mathcal{F}_{S_{n}}$. Therefore, the inter-sampling time $T_{n}$ is the \emph{first exit time} of $\widetilde{W}_t$ starting from zero and exiting from the interval $(-\alpha_{n},\beta_{n})$, i.e.,
\begin{equation}
    T_{n} = \inf\{t \geq 0: \widetilde{W}_{t} \notin (-\alpha_{n}, \beta_{n})\}, \label{eq:exit-time}
\end{equation}
where $\alpha_{0} = \beta_{0} = \delta$, and for $n \geq 1$,
\begin{equation}
    (\alpha_{n}, \beta_{n}) = \begin{cases}
        \left(\delta, (n+1)\delta \right), & D_{n}=-1,\\
        \left((n+1) \delta, \delta \right), & D_{n}=+1.
    \end{cases} \label{eq:interval}
\end{equation}

\begin{remark}
The strong Markov property implies that the two-dimensional state $(n, D_{n})$ is sufficient to characterize the sampler. The history-dependent variable $\Xi_{n}$ is required only to decode the innovation value. 
\end{remark}

In contrast, under deviation-based sampling, each inter-sampling time is a first exit time of $\widetilde{W}_{t}$ from a \emph{fixed} interval $(-\delta, \delta)$. That is, $\alpha_{n} = \beta_{n} = \delta$ for all $n$, and
\begin{equation}
    T^{\mathrm{dev}}_{n} = \inf\{t\geq 0: \widetilde{W}_{t} \notin (-\delta, \delta)\}.
\end{equation}

\begin{figure}[t!]
    \centering
    \subfloat[deviation-based sampling]{
    \includegraphics[width=0.65\linewidth]{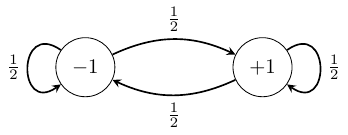}
    \label{fig:binary_chain_deviation}
    }
    \\
    \subfloat[innovation-based sampling]{
    \includegraphics[width=0.7\linewidth]{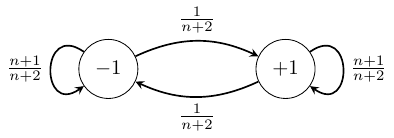}
    \label{fig:binary_chain_innovation}
    }
    \\
    \caption{Binary chain representation of the innovation directions $(D_{n})_{n\geq 1}$.}
    \label{fig:binary_chain}
\end{figure}

\subsection{Binary Chain Representation of Innovation Directions}
This section presents binary chain representations of the direction processes under innovation-based and deviation-based sampling, as illustrated in Fig.~\ref{fig:binary_chain}. Innovation-based sampling induces an increasingly persistent direction process, with a reversal probability that vanishes as $n\to\infty$, whereas deviation-based sampling yields independent, equiprobable directions.

\begin{theorem}\label{lemma:evolution}
The direction process $(D_n)_{n\geq 1}$ forms a time-inhomogeneous binary Markov chain with a transition probability matrix
\begin{equation}
    P_{n} = \begin{bmatrix}
        \frac{n+1}{n+2} & \frac{1}{n+2} \\
        \frac{1}{n+2} & \frac{n+1}{n+2}
    \end{bmatrix},
\end{equation}
where, for $d\in\{-1,+1\}$, $\mathbf{P}(D_{n+1}=d| D_n=d) = \frac{n+1}{n+2}$ is the direction holding probability, and $\mathbf{P}(D_{n+1}=-d| D_n=d) = \frac{1}{n+2}$ is the direction reversal probability.
\end{theorem}
\begin{IEEEproof}
For the Wiener process $\widetilde{W}_t$ starting from zero, the first exit probabilities from $(-\alpha_{n}, \beta_{n})$ are (see~\eqref{eq:app-exit-probabilities})
\begin{equation*}
    \mathbf{P}(\widetilde{W}_{T_{n}}=\beta_{n}) = 
    \frac{\alpha_{n}}{\alpha_{n}+\beta_{n}}, \,\,\,
    \mathbf{P}(\widetilde{W}_{T_{n}}=-\alpha_{n})
    = \frac{\beta_{n}}{\alpha_{n}+\beta_{n}}.
\end{equation*}
For $n\geq1$, suppose that $D_n=+1$. Then $\alpha_n=(n+1)\delta$ and $\beta_n=\delta$. Hence, the direction holding probability is
\begin{align*}
    \mathbf{P}(D_{n+1}= +1 {\,}| {\,} D_{n} = + 1) 
    &= \mathbf{P}(\widetilde{W}_{T_{n}} = \beta_{n} {\,}| {\,} D_{n} = + 1) \\
    &= \frac{(n+1)\delta}{\delta + (n+1)\delta} = \frac{n+1}{n+2},
\end{align*}
while the direction reversal probability is
\begin{align*}
    \mathbf{P}(D_{n+1}=-1 {\,}| {\,} D_{n} = +1) 
    &= \mathbf{P}(\widetilde{W}_{T_{n}} = -\alpha_{n} {\,}| {\,} D_{n} = + 1) \\
    &= \frac{\delta}{\delta + (n+1)\delta} = \frac{1}{n+2}.
\end{align*}
By symmetry, the same holding and reversal probabilities hold when $D_n=-1$. This completes the proof.
\end{IEEEproof}

% Theorem~\ref{lemma:evolution} shows that innovations become increasingly persistent in direction. 

\subsection{Independence and Limit Theorems for Direction Reversals}
Since the reversal probability 
\begin{equation}
    \mathbf{P}(I_{n} = 1|D_{1}, D_{2}, \ldots, D_{n-1}) = \frac{1}{n+1}
\end{equation}
does not depend on the previous directions, the reversal indicators $(I_{n})_{n\geq 2}$ are \emph{mutually independent}, with
\begin{equation}
    I_{n} \sim \mathsf{Bernoulli}(\frac{1}{n+1}), \quad n \geq 2. \label{eq:reversal-indicator-distribution}
\end{equation}

Under innovation-based sampling, the reversal surprisal is $\log_2(n+1)$, which increases logarithmically with the number of innovations. \emph{Since reversals become increasingly rare, their occurrence becomes increasingly informative.}

\begin{remark}
Let $(\xi_{n})_{n\geq 1}$ denote the upper-record indicators of an i.i.d. sequence with a continuous distribution. A classical result of R\'enyi~\cite{renyi1962theorie} states that these indicators are mutually independent and satisfy $\xi_{n}\sim\mathsf{Bernoulli}(1/n)$. Thus, the reversal indicators form a shifted sequence of classical record indicators, and moments and limit theorems for the reversal count follow directly from record theory~\cite{arnold2011records}. The main consequences are presented below.
\end{remark}

\begin{lemma}
The expected number of reversals is given by
\begin{equation}
    \mathbb{E}[V_{n}]
    = 
    H_{n+1}-\frac{3}{2}
    =
    \log(n)+\mathcal{O}(1),
    \label{eq:expected_reversals}
\end{equation}
where $H_{n} = \sum_{k=1}^{n}1/k$ denotes the $n$th harmonic number. Moreover, as $n\to \infty$,
\begin{equation}
    \frac{V_{n}}{\log(n)} \xrightarrow{a.s.}
    1, \qquad \frac{
        V_{n}-\mathbb{E}[V_{n}]
    }{
        \sqrt{\operatorname{Var}(V_{n})}
    }
    \xrightarrow{d}
    \mathcal{N}(0,1).
    \label{eq:reversal-slln}
\end{equation}
where $\xrightarrow{a.s.}$ denotes almost sure convergence, and $\xrightarrow{d}$ denotes convergence in distribution.
\end{lemma}

Thus, the number of reversals grows logarithmically in the number of innovations almost surely, while its standardized fluctuations are asymptotically Gaussian.

% \begin{remark}
% Under innovation-based sampling, the reversal surprisal is $\log_2(n+1)$, which increases logarithmically with the number of innovations. \emph{Since reversals become increasingly rare, their occurrence becomes increasingly informative.}
% \end{remark}

In contrast, for deviation-based sampling, the successive increments are $W_{S_{n+1}} - W_{S_{n}} \in \{-\delta, \delta\}$, with independent equiprobable signs by the strong Markov property. Hence, the reversal probability is $\mathbf{P}(I_{n} = 1) = 1/2$, and each reversal has a surprisal of one bit. Moreover, the expected number of reversals grows linearly as
\begin{equation}
    \mathbb{E}[V^{\mathrm{dev}}_n] = \frac{n-1}{2}.
\end{equation}
Since reversals occur frequently, they are less informative.

\subsection{Distributions of the Inter-Sampling and Sampling Times}
We next characterize the distributions of the inter-sampling times $(T_{n})_{n\geq 0}$ and the sampling times $(S_{n})_{n\geq 1}$ and derive their moments. We show that the innovation-based sampler becomes \emph{less active} with the number of innovations.

\begin{theorem}\label{thm:sampling_times}
For innovation-based sampling of a standard Wiener process,
the distribution of the $n$th inter-sampling time $T_n$ is characterized by
the Laplace transform
\begin{equation}
    \mathcal{L}_{T_n}(s) 
    =
    \frac{
    \cosh\left(\frac{n\delta}{2}\sqrt{2s}\right)
    }{
    \cosh\left(\frac{(n+2)\delta}{2}\sqrt{2s}\right)
    }, \quad s\geq 0. \label{eq:LT_Tn}
\end{equation}
Moreover, $(T_n)_{n\geq 0}$ are mutually independent, with
\begin{equation}
\label{eq:moments_Tn}
    \mathbb{E}[T_n]
    =
    (n+1)\delta^2,\,\,
    \operatorname{Var}(T_n)
    =
    \frac{(n+1)\left((n+1)^2+1\right)}{3}\delta^4.
\end{equation}
The $n$th sampling time $S_n$ is characterized by
\begin{equation}
\label{eq:LT_Sn}
     \mathcal{L}_{S_n}(s) 
    =
    \frac{
    \cosh\left(\frac{\delta}{2}\sqrt{2s}\right)
    }{
    \cosh\left(\frac{n\delta}{2}\sqrt{2s}\right)
    \cosh\left(\frac{(n+1)\delta}{2}\sqrt{2s}\right)}, \quad s\geq 0.
\end{equation}
and the mean and variance are 
\begin{equation}
\label{eq:moments_Sn}
    \mathbb{E}[S_n]
    =
    \frac{n(n+1)}{2}\delta^2,
    \operatorname{Var}(S_n)
    =
    \frac{n(n+1)(n(n+1)+2)}{12}\delta^4.
\end{equation}
\end{theorem}
\begin{IEEEproof}
Recall that, conditioned on $\mathcal{F}_{S_n}$, $T_n$ is the first exit time of the shifted Wiener process $\widetilde{W}_t$ from $(-\alpha_n,\beta_n)$. By~\eqref{eq:app-exit-laplace}, the conditional Laplace transform of $T_n$ is
\begin{equation*}
    \mathcal{L}_{T_n}(s) 
    :=  \mathbb{E}[e^{-s T_{n}}] 
    = \frac{\cosh\left(\frac{\alpha_{n} - \beta_{n}}{2}\sqrt{2s}\right)}{\cosh\left(\frac{\alpha_{n} + \beta_{n}}{2}\sqrt{2s} \right)}, \quad s \geq 0.
\end{equation*}
Since $\{\alpha_n,\beta_n\} = \{\delta,(n+1)\delta\}$, we have $\alpha_n+\beta_n=(n+2)\delta$ and $|\alpha_n-\beta_n|=n\delta$. Because $\cosh(\cdot)$ is even, the above conditional transform reduces to
\begin{equation*}
    \mathcal{L}_{T_n}(s)
    =
    \frac{
    \cosh\left(\frac{n\delta}{2}\sqrt{2s}\right)
    }{
    \cosh\left(\frac{(n+2)\delta}{2}\sqrt{2s}\right)
    } = \frac{\cosh(n\theta)}{\cosh((n+2)\theta)},
\end{equation*}
where $\theta := \frac{\delta}{2}\sqrt{2s}$. This yields~\eqref{eq:LT_Tn}. Moreover, the right-hand side (RHS) does not depend on
$\mathcal{F}_{S_n}$, and hence
\begin{equation*}
    \mathcal{L}_{T_n}(s) 
    = \mathbb{E}[e^{-s T_{n}}]
    = \mathbb{E}[e^{-s T_{n}} {\,}|{\,} \mathcal{F}_{S_n}].
\end{equation*}
Since 
\begin{equation*}
    \sigma(T_{0}, T_{1}, \ldots, T_{n-1}) \subseteq \mathcal{F}_{S_{n}},
\end{equation*}
$T_n$ is independent of $(T_0,T_1,\ldots,T_{n-1})$. This implies that $(T_n)_{n\geq0}$ are mutually independent.

We next derive the moments of $T_n$. The $k$th derivative of $\mathcal{L}_{T_{n}}(s)$ is given by
\begin{equation*}
    \mathcal{L}_{T_n}^{(k)}(s)
    =
    (-1)^k
    \mathbb{E}[T_n^k e^{-sT_n}],
\end{equation*}
and, in particular,
\begin{equation*}
    \mathcal{L}_{T_n}'(0)
    =
    -\mathbb{E}[T_n],
    \quad
    \mathcal{L}_{T_n}''(0)
    =
    \mathbb{E}[T_n^2].
\end{equation*}
Using the Taylor expansion $\cosh(\sqrt{2s}\,x) = 1+s x^2+\frac{s^2x^4}{6}+\mathcal{O}(s^3)$,
we obtain
\begin{equation*}
    \mathcal{L}_{T_n}(s)
    =
    1-\alpha_n\beta_n s
    +
    \frac{
    \alpha_n\beta_n
    \left(
    \alpha_n^2+3\alpha_n\beta_n+\beta_n^2
    \right)
    }{6}s^2
    +\mathcal{O}(s^3).
\end{equation*}
Therefore,
\begin{equation*}
    \mathbb{E}[T_n] =
    \alpha_n \beta_n, \quad \mathbb{E}[T_n^2] =
    \frac{\alpha_n\beta_n}{3}
    \left(
    \alpha_n^2+3\alpha_n \beta_n+\beta_n^2
    \right).
\end{equation*}
Substituting
$\{\alpha_n,\beta_n\}=\{\delta,(n+1)\delta\}$ gives
\begin{equation*}
    \mathbb{E}[T_n] = (n+1)\delta^2, \quad
    \mathbb{E}[T_n^2] =
    \frac{(n+1)(n^2+5n+5)}{3}\delta^4.
\end{equation*}
Hence,
\begin{equation*}
    \operatorname{Var}(T_n)
    =
    \mathbb{E}[T_n^2]-\mathbb{E}[T_n]^2
    =
    \frac{(n+1)\left((n+1)^2+1\right)}{3}\delta^4,
\end{equation*}
which proves~\eqref{eq:moments_Tn}.

We now turn to the sampling time. Since $S_n=\sum_{k=0}^{n-1}T_k$ and the inter-sampling times are mutually independent, its Laplace transform factorizes as
\begin{align*}
    \mathcal{L}_{S_n}(s)
    &=
    \prod_{k=0}^{n-1}\mathcal{L}_{T_k}(s)\\
    &=
    \prod_{k=0}^{n-1}
    \frac{
    \cosh\left(k\theta\right)
    }{\cosh((k+2)\theta)
    }\\
    &=
    \frac{\cosh(0)\cosh(\theta)
    \cosh(2\theta)\cdots \cosh((n-1)\theta)}{\cosh(2\theta)\cosh(3\theta)\cdots \cosh(n\theta)\cosh((n+1)\theta)} \\
    &= 
    \frac{\cosh(\theta)}{\cosh(n\theta)\cosh((n+1)\theta)
    },
\end{align*}
which proves~\eqref{eq:LT_Sn}. Finally,
\begin{align*}
    \mathbb{E}[S_n]
    &=
    \sum_{k=0}^{n-1}\mathbb{E}[T_k]
    =
    \frac{n(n+1)}{2}\delta^2,
\end{align*}
and, by independence,
\begin{align*}
    \operatorname{Var}(S_n)
    &=
    \sum_{k=0}^{n-1}\mathrm{Var}(T_k)
    =
    \frac{\delta^4}{3}
    \sum_{j=1}^{n}(j^3+j)\\
    &=
    \frac{n(n+1)\left(n(n+1)+2\right)}{12}\delta^4.
\end{align*}
This completes the proof.
\end{IEEEproof}

The above result reveals a key temporal property of the innovation-based sampler. The expected inter-sampling time grows linearly with the number of innovations, while the expected sampling time grows quadratically. Thus, innovations become progressively sparser as the range of previously communicated values expands. 

By contrast, under deviation-based sampling, every inter-sampling time is a first exit time of $\widetilde{W}_t$ from the fixed interval $(-\delta,\delta)$. As a result, the inter-sampling times $(T_n^{\mathrm{dev}})_{n\geq0}$ are i.i.d., with Laplace transform
\begin{equation}
    \mathbb{E}[e^{-sT_n^{\mathrm{dev}}}]
    =
    \frac{1}{\cosh(\delta\sqrt{2s})},
    \quad s\geq 0,
\end{equation}
and mean
\begin{equation}
    \mathbb{E}[T_n^{\mathrm{dev}}] =
    \delta^2, \quad n=0,1,2,\ldots.
\end{equation}
Thus, the sampling activity does not decrease as more samples are generated.
We note that these results coincide with the innovation-based results at $n=0$.

\subsection{Expected and Asymptotic Sampling Rates}
We next characterize how frequently innovations are generated over time and show that the innovation-based sampler becomes \emph{asymptotically silent}.

\begin{theorem}
Under innovation-based sampling of a Wiener process, the expected number of innovations is given by
\begin{equation}
    \mathbb{E}[N_t]
    =
    4\sum_{k=1}^{\infty}
    Q\left(\frac{k\delta}{\sqrt{t}}\right),
    \label{eq:EN_exact}
\end{equation}
where $Q(x)$ is the Gaussian tail function. Moreover,
\begin{equation}
    \mathbb{E}[N_{t}] = \frac{1}{\delta}\sqrt{\frac{8t}{\pi}} + \mathcal{O}(1), \quad t \to \infty.
\end{equation}
Consequently, the expected sampling rate satisfies
\begin{equation}
    R_t = \frac{1}{\delta}\sqrt{\frac{8}{\pi t}} + \mathcal{O}(t^{-1}),
    \quad t \to\infty,
\end{equation}
and the asymptotic sampling rate is $R_{\infty} = 0$.
\end{theorem}
\begin{IEEEproof}
Recall the running extrema $\overline{M}_t$ and $\underline{m}_t$ in
Definition~\ref{def:extrema}. By Lemma~\ref{lemma:spread}, each increase of
$\overline{M}_t$ by $\delta$ generates one upper innovation, while each
decrease of $\underline{m}_t$ by $\delta$ generates one lower innovation.
Hence, the numbers of upper and lower innovations generated by time $t$ are
\begin{equation*}
    N_t^{+}
    =
    \left\lfloor\frac{\overline{M}_t}{\delta}\right\rfloor,
    \qquad
    N_t^{-}
    =
    \left\lfloor\frac{-\underline{m}_t}{\delta}\right\rfloor,
\end{equation*}
respectively. Since every innovation is either upper or lower, 
\begin{equation*}
    N_t = N_t^{+} + N_t^{-}.
\end{equation*}

We first derive the exact expectation. For any nonnegative random variable
$X$, $\mathbb{E}[\lfloor X\rfloor] = \sum_{k=1}^{\infty}\mathbf{P}(X\geq k)$. Therefore,
\begin{equation*}
    \mathbb{E}[N_t] = \sum_{k=1}^{\infty}
    \mathbf{P}(\overline{M}_t\geq k\delta) +
    \sum_{k=1}^{\infty}
    \mathbf{P}(-\underline{m}_t\geq k\delta).
\end{equation*}
By symmetry of the Wiener process,
\begin{equation*}
    \mathbf{P}(-\underline{m}_t\geq x)
    =
    \mathbf{P}(\overline{M}_t\geq x).
\end{equation*}
Moreover, applying the reflection principle
(see Lemma~\ref{lem:reflection-principle} in the Appendix), we obtain
\begin{equation*}
    \mathbf{P}(\overline{M}_t\geq x)
    =
    2\mathbf{P}(W_t\geq x)
    =
    2Q\left(\frac{x}{\sqrt{t}}\right),
    \quad x\geq0.
\end{equation*}
Substituting $x=k\delta$ yields
\begin{equation*}
    \mathbb{E}[N_t]
    =
    4\sum_{k=1}^{\infty}
    Q\left(\frac{k\delta}{\sqrt{t}}\right).
\end{equation*}
which proves~\eqref{eq:EN_exact}.

We next derive the asymptotic growth of $\mathbb{E}[N_t]$. Using $x-1<\lfloor x\rfloor\leq x$ gives
\begin{equation*}
    \frac{\overline{M}_t-\underline{m}_t}{\delta}-2
    <
    N_t
    \leq
    \frac{\overline{M}_t-\underline{m}_t}{\delta}.
\end{equation*}
Taking expectations yields
\begin{equation*}
    \frac{
    \mathbb{E}[\overline{M}_t]
    +
    \mathbb{E}[-\underline{m}_t]
    }{\delta}
    -2
    <
    \mathbb{E}[N_t]
    \leq
    \frac{
    \mathbb{E}[\overline{M}_t]
    +
    \mathbb{E}[-\underline{m}_t]
    }{\delta}.
\end{equation*}
By Lemma~\ref{lem:reflection-principle}, we have
\begin{equation*}
    \mathbb{E}[\overline{M}_t]
    =
    \mathbb{E}[-\underline{m}_t]
    =
    \sqrt{\frac{2t}{\pi}}.
\end{equation*}
Consequently,
\begin{equation*}
    \frac{1}{\delta}\sqrt{\frac{8t}{\pi}}-2
    <
    \mathbb{E}[N_t]
    \leq
    \frac{1}{\delta}\sqrt{\frac{8t}{\pi}},
\end{equation*}
which implies
\begin{equation*}
    \mathbb{E}[N_t]
    =
    \frac{1}{\delta}\sqrt{\frac{8t}{\pi}}
    +
    \mathcal{O}(1).
\end{equation*}
Finally, dividing by $t$ gives
\begin{equation*}
    R_t = \frac{\mathbb{E}[N_t]}{t}
    =
    \frac{1}{\delta}\sqrt{\frac{8}{\pi t}}
    +
    \mathcal{O}(t^{-1}),
\end{equation*}
and therefore
\begin{equation*}
    R_{\infty} = \lim_{t\to\infty}R_t = 0.
\end{equation*}
This completes the proof.
\end{IEEEproof}

\begin{figure}[t!]
\centering
\includegraphics[width=\linewidth]{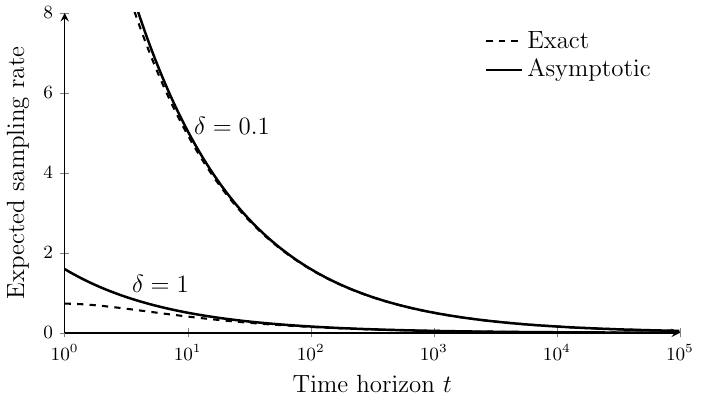}
\caption{Asymptotic and expected sampling rates.}
\label{fig:sampling_rate}
\end{figure}

The theorem shows that the expected number of innovations grows on the order of $\sqrt{t}$, while the expected sampling rate decays on the order of $t^{-1/2}$ and vanishes asymptotically, as illustrated in Fig.~\ref{fig:sampling_rate}. Under deviation-based sampling, however, the asymptotic sampling rate is the positive constant $1/\delta^{2}$, which can become excessively large when $\delta$ is small.

\begin{table*}[t!]
\centering
\caption{Comparison of the Proposed Estimators Under Innovation-Based Sampling}
\label{tab:estimator_comparison}
\renewcommand{\arraystretch}{1.5}
\setlength{\tabcolsep}{3pt}
\begin{tabular}{lccccc}
\hline
\textbf{Property}
& \textbf{MMSE}
& \textbf{EXP}
& \textbf{OAC}
& \textbf{LSI}
& \textbf{ZOH}
\\
\hline

Estimation rule
& $W_{S_n} \!\!+\! f_n(U_t)$
& $W_{S_n} \!\!+\! \frac{\beta_n-\alpha_n}{2}(1 \!-\! e^{-r_nU_t})$
& $W_{S_n} \!\!+\! c_n^* \!+\! b_n^*U_t$
& $W_{S_n} \!\!+\! \frac{\beta_n-\alpha_n}{2}$
& $W_{S_n}$
\\

Silence-aware
& Yes
& Yes
& Yes
& Yes
& No
\\

AoI-aware
& Yes
& Yes
& Yes
& No
& No
\\

Average MSE
& --
& $\frac{h(\rho)}{2}n^2\delta^2 \!+\! \mathcal{O}(n\delta^2)$
& $\frac{3}{140}n^2\delta^2 \!+\! \mathcal{O}(n\delta^2)$
& $\frac{n^2+n+2}{24}\delta^2$
& $\frac{3n^2-n+4}{36}\delta^2$
\\

Asymptotic gain
& $82.14\%$
& $81.85\%$
& $74.29\%$
& $50\%$
& $0$
\\
\hline
\end{tabular}

\smallskip
\begin{minipage}{0.75\textwidth}
\scriptsize
For EXP, $r_n=\rho /(\alpha_n + \beta_n)^2$ with $\rho^* = 16.82$ and $h(\rho)$ given by~\eqref{eq:hrho}. For OAC, $(c_n^*,b_n^*)$ are given
in~\eqref{eq:oac-coeffs}.
\end{minipage}
\end{table*}

\section{Estimators Under Innovation-Based Sampling}\label{sec:estimation-wiener-process}
This section studies estimation under innovation-based sampling. Section~\ref{sec:mmse} derives the MMSE estimator by exploiting the information conveyed by the absence of new innovations. Since its correction term involves an infinite series, Section~\ref{sec:affine-age} introduces a tractable class of affine-age estimators to approximate the MMSE estimator and characterizes their performance. This class includes the ZOH, long-silence (LSI), and optimal affine-correction (OAC) estimators. Finally, Section~\ref{sec:exp-age} develops a near-optimal exponential-age (EXP) estimator whose correction approaches the MMSE long-silence limit. The main results are summarized in Table~\ref{tab:estimator_comparison}.

\subsection{MMSE Estimator and Limiting Behavior}\label{sec:mmse}
The following theorem characterizes the MMSE estimator and reveals a key distinction between innovation-based and deviation-based sampling.
\begin{itemize}
    \item Under innovation-based sampling, \emph{silence is informative, and the MMSE estimate evolves with the AoI}. Thus, information is conveyed through the timing of innovations.\footnote{This perspective is related to the concept of a timing channel~\cite{anantharam1996bits}, in which information is conveyed through timing rather than solely through content.}
    \item Under deviation-based sampling, however, the MMSE estimator reduces to the silence-ignorant ZOH estimator. 
\end{itemize}

\begin{theorem}\label{thm:mmse}
The MMSE estimator under innovation-based sampling is 
\begin{equation}
    \hat{W}_{t} = W_{S_{n}} + f_{n}(U_{t}), \quad t \in [S_{n}, S_{n+1}), \label{eq:exact-mmse-estimator}
\end{equation}
where $f_{n}(u)$ is an AoI-aware correction term given by
\begin{equation}
    f_{n}(u) = \frac{\int_{-\alpha_{n}}^{\beta_{n}} x {\,} p_{n}(u, x)\,dx}{\int_{-\alpha_{n}}^{\beta_{n}} p_{n}(u, x)\,dx}. \label{eq:mmse-correction}
\end{equation}
Here, $p_{n}(u, x)$ is the transition density of a standard Wiener process starting from zero and killed upon exiting $(-\alpha_{n}, \beta_{n})$, given by
\begin{align}
    p_{n}(u, x) &= \frac{2}{\alpha_{n} + \beta_{n}}\sum_{k=1}^{\infty} \sin\left(\frac{k\pi \alpha_{n}}{\alpha_{n} + \beta_{n}}\right) \sin\left(\frac{k\pi (x+\alpha_{n})}{\alpha_{n} + \beta_{n}}\right)\notag\\
    & \qquad \qquad \qquad \times \exp\left(-\frac{k^{2}\pi^{2}u}{2(\alpha_{n} + \beta_{n})^{2}}\right). \label{eq:density}
\end{align}
Moreover, under deviation-based sampling, the MMSE estimator reduces to the ZOH estimator
\begin{equation}
    \hat{W}_{t} = W_{S_{n}}, \quad t \in [S_{n}, S_{n+1}).
\end{equation}
\end{theorem}
\begin{IEEEproof}
Recall the shifted Wiener process~\eqref{eq:shifted-process} and the first exit time~\eqref{eq:exit-time} from the $n$th interval $(-\alpha_{n}, \beta_{n})$. We may write
\begin{equation*}
    W_{t} = W_{S_{n}} + \widetilde{W}_{U_{t}}, \quad t \in [S_{n}, S_{n+1}).
\end{equation*}
The MMSE estimator therefore decomposes as
\begin{equation}
    \hat{W}_{t} 
    = \underbrace{\mathbb{E}[W_{S_{n}} {\,}|{\,} \mathcal{H}_{n}, T_{n}> U_{t}]}_{= W_{S_{n}}} 
    + \underbrace{\mathbb{E}[\widetilde{W}_{U_{t}} {\,}|{\,} \mathcal{H}_{n}, T_{n}> U_{t}]}_{=: f_{n}(U_{t})}.
\end{equation}
The first term equals the last received sample because $W_{S_n}$ is $\mathcal{H}_n$-measurable. The second term, $f_n(U_t)$, is the correction inferred from silence. For $u>0$,
\begin{equation*}
    f_n(u)
    =
    \mathbb{E}[\widetilde{W}_u
    {\,}|{\,} \mathcal{H}_n, T_n>u]
    =
    \frac{
        \mathbb{E}[\widetilde{W}_u
        \mathbb{I}\{T_n>u\} {\,}|{\,} \mathcal{H}_n]
    }{
        \mathbf{P}(T_n>u {\,}|{\,} \mathcal{H}_n)
    },
\end{equation*}
where
\begin{equation}
    \mathbf{P}(T_n>u {\,}|{\,} \mathcal{H}_n) = \int_{-\alpha_{n}}^{\beta_{n}} p_{n}(u, x)\,dx, \label{eq:survival probability}
\end{equation}
is the survival probability, i.e., the probability that the process remains in $(-\alpha_n, \beta_n)$ throughout $[0, u]$, and $p_{n}(u, x)$ is the killed transition density of $\widetilde{W}_{u}$ starting from zero and killed upon exiting this interval. The expression in~\eqref{eq:density} follows directly from~\eqref{eq:app-killed-density-general} in Appendix~\ref{app:first-exit-time}. Then, we may write
\begin{equation*}
    f_{n}(u) = \frac{\int_{-\alpha_{n}}^{\beta_{n}} x {\,} p_{n}(u, x)\,dx}{\int_{-\alpha_{n}}^{\beta_{n}} p_{n}(u, x)\,dx}.
\end{equation*}

For deviation-based sampling, the continuation region after every sample is symmetric, i.e., $(-\alpha_{n}, \beta_{n}) = (-\delta, \delta)$. The killed transition density~\eqref{eq:density} becomes
\begin{align}
    &p^{\mathrm{dev}}_{n}(u, x) \notag\\
    &= \frac{1}{\delta} \sum_{k=1}^{\infty} \sin\left(\frac{k\pi}{2}\right) \sin\left(\frac{k\pi (x+\delta)}{2\delta}\right) \exp\left(-\frac{k^{2}\pi^{2}u}{8 \delta^{2}}\right) \notag\\
    &= \frac{1}{\delta}
    \sum_{j=0}^{\infty}
    \cos\left(\frac{(2j+1)\pi x}{2\delta}\right)
    \exp\left(-\frac{(2j+1)^2\pi^2u}{8\delta^2}\right).
\end{align}
Since each cosine term is even in $x$, this density satisfies
\begin{equation*}
    p^{\mathrm{dev}}_{n}(u, x) 
    = p^{\mathrm{dev}}_{n}(u, -x).
\end{equation*}
Hence,
\begin{equation}
    f^{\mathrm{dev}}_n(u)
    =
    \frac{
        \int_{-\delta}^{\delta}
        x p^{\mathrm{dev}}_n(u,x)\,dx
    }{
        \int_{-\delta}^{\delta}
        p^{\mathrm{dev}}_n(u,x)\,dx
    }
    =0,
\end{equation}
and the MMSE estimator reduces to $\hat{W}_{t} = W_{S_{n}}$.
\end{IEEEproof}

The next lemma characterizes the long-silence limit of the MMSE correction. It shows that, as AoI tends to infinity, the correction approaches the midpoint of the continuation interval. This result motivates the LSI estimator introduced in Section~\ref{sec:affine-age}.

\begin{lemma}\label{lemma:LS}
For each $n\geq 1$, the MMSE correction satisfies
\begin{equation}
    \lim_{u \to \infty} f_{n}(u) 
    = 
    \frac{\beta_{n} - \alpha_{n}}{2} 
    = 
    - D_{n} \frac{n\delta}{2}.
\end{equation}
\end{lemma}
\begin{IEEEproof}
Fix $n$ and write $L_n=\alpha_n+\beta_n$. The eigenvalues in the killed transition density~\eqref{eq:density} are
\begin{equation*}
    \lambda_k=\frac{k^2\pi^2}{2L_n^2},
    \quad k\geq1,
\end{equation*}
Let
\begin{equation*}
    \xi_n(x)
    =
    \frac{2}{L_n}
    \sin\left(\frac{\pi\alpha_n}{L_n}\right)
    \sin\left(\frac{\pi(x+\alpha_n)}{L_n}\right).
\end{equation*}
By~\eqref{eq:density},
\begin{equation*}
    \sup_{x\in[-\alpha_n,\beta_n]}
    \left|e^{\lambda_1u}p_n(u,x)-\xi_n(x)\right|
    \leq
    \frac{2}{L_n}
    \sum_{k=2}^{\infty}e^{-(\lambda_k-\lambda_1)u}.
\end{equation*}
Since $\lambda_k>\lambda_1$ for $k\geq2$, the series on the RHS tends to zero as $u\to\infty$ by dominated convergence. Hence,
$e^{\lambda_{1}u}p_n(u,x)$ converges uniformly to $\xi_{n}(x)$ on $[-\alpha_{n},\beta_{n}]$.

Multiplying the numerator and denominator
of~\eqref{eq:mmse-correction} by $e^{\lambda_1u}$ and
passing the limit through the integrals therefore gives
\begin{align*}
    \lim_{u\to\infty}f_n(u)
    &=
    \frac{
        \int_{-\alpha_n}^{\beta_n}x\xi_n(x)\,dx
    }{
        \int_{-\alpha_n}^{\beta_n}\xi_n(x)\,dx
    } \\
    &=
    \frac{
        \int_{-\alpha_n}^{\beta_n}
        x\sin\left(\frac{\pi(x+\alpha_n)}{L_n}\right) dx
    }{
        \int_{-\alpha_n}^{\beta_n}
        \sin\left(\frac{\pi(x+\alpha_n)}{L_n}\right) dx
    }.
\end{align*}
Since the sine function is symmetric about the midpoint $\frac{\beta_{n} - \alpha_{n}}{2}$, its normalized mean equals this midpoint. Thus,
\begin{equation*}
    \lim_{u\to \infty} f_{n}(u) = \frac{\beta_{n} - \alpha_{n}}{2} = -D_n \frac{n\delta}{2},
\end{equation*}
which completes the proof.
\end{IEEEproof}

The next lemma quantifies the MSE reduction relative to ZOH achieved by exploiting silence. 

\begin{lemma}\label{lem:mmse-mse}
The MSE of the MMSE estimator~\eqref{eq:exact-mmse-estimator} over the $n$th sampling interval is
\begin{equation}
    \mathsf{MSE}_n^{*}
    =
    \mathsf{MSE}_n^{\mathrm{ZOH}}
    -
    \int_0^{\infty}
    f_n^2(u)
    \int_{-\alpha_n}^{\beta_n}
    p_n(u,x)\,dx\,du.
\end{equation}
where $f_n(u)$ is given by~\eqref{eq:mmse-correction} and $p_n(u,x)$ is given by~\eqref{eq:density}.
\end{lemma}
\begin{IEEEproof}
Recall that
\begin{equation*}
    W_t=W_{S_n} + \widetilde{W}_{U_t},
    \quad
    \hat W_t=W_{S_n}+f_n(U_t),
\end{equation*}
for $t\in[S_n,S_{n+1})$. Hence,
\begin{align*}
    \mathsf{MSE}_n^{*}
    &=
    \mathbb{E}\left[
        \int_0^{T_n}
        \bigl(\widetilde W_u-f_n(u)\bigr)^2\,du
    \right] \\
    &=
    \int_0^\infty
    \int_{-\alpha_n}^{\beta_n}
        \bigl(x-f_n(u)\bigr)^2
        p_n(u,x)\,dx\,du. \\
    &=
    \int_0^\infty
    \int_{-\alpha_n}^{\beta_n}
        x^2p_n(u,x)\,dx\,du \\
    &\quad
    -2\int_0^\infty
        f_n(u)
        \int_{-\alpha_n}^{\beta_n}
        xp_n(u,x)\,dx\,du \\
    &\quad
    +\int_0^\infty
        f_n^2(u)
        \int_{-\alpha_n}^{\beta_n}
        p_n(u,x)\,dx\,du.
\end{align*}
The first term on the RHS satisfies
\begin{equation*}
    \int_0^\infty 
    \int_{-\alpha_n}^{\beta_n}
        x^2p_n(u,x)\,dx\,du
    = \mathbb{E} \left[
        \int_0^{T_n} \widetilde{W}_u^2\,du
    \right] = \mathsf{MSE}_n^{\mathrm{ZOH}}.
\end{equation*}
By~\eqref{eq:mmse-correction}, 
\begin{equation*}
    \int_{-\alpha_n}^{\beta_n}
        xp_n(u,x)\,dx
    =
    f_n(u)
    \int_{-\alpha_n}^{\beta_n}
        p_n(u,x)\,dx.
\end{equation*}
Therefore,
\begin{equation*}
    \mathsf{MSE}_n^{*}
    =
    \mathsf{MSE}_n^{\mathrm{ZOH}} -
    \int_0^\infty
        f_n^2(u)
        \int_{-\alpha_n}^{\beta_n}
        p_n(u,x)\,dx\,du,
\end{equation*}
which yields the desired result.
\end{IEEEproof}

\subsection{Affine-Age Estimators}\label{sec:affine-age}
The exact MMSE correction $f_n(u)$ depends on the AoI and generally requires evaluating an infinite series, which may limit its practicality in resource-constrained sensor networks. We therefore consider a tractable class of affine-age estimators, which admit closed-form MSE expressions.

\begin{definition}\label{def:affine-age-estimator}
An \emph{affine-age estimator} has the form
\begin{equation}
    \hat{W}_{t} = W_{S_n}+c_n + b_n U_t,
    \quad t \in [S_n,S_{n+1}),
    \label{eq:affine-age-estimator}
\end{equation}
where $c_{n}$ and $b_{n}$ are $\mathcal{H}_{n}$-measurable and remain fixed during the $n$th sampling interval.
\end{definition}

We first characterize the estimation error of this class.

\begin{theorem}\label{thm:affine-age-estimator}
The MSE of the affine-age estimator~\eqref{eq:affine-age-estimator} over the $n$th sampling interval is
\begin{align}
    \mathsf{MSE}^{\mathrm{AFF}}_{n}(c_n, b_n) 
    &= \mathcal{E}_{n, 0} + \mathcal{E}_{n, 1} c_n + \mathcal{E}_{n, 2} c_n^2 \notag\\
    &\quad + \mathcal{E}_{n, 3} c_n b_n + \mathcal{E}_{n, 4} b_n + \mathcal{E}_{n, 5} b_n^2, \label{eq:affine-age-mse}
\end{align}
where
\begin{subequations}
\begin{align}
    \mathcal{E}_{n,0}
    &=
    \frac{\alpha_n\beta_n}{6}
    \left(
        \alpha_n^2-\alpha_n\beta_n+\beta_n^2
    \right),
    \label{eq:affine-E0}\\
    \mathcal{E}_{n,1}
    &=
    \frac{-2\alpha_n\beta_n}{3}
    (\beta_n-\alpha_n),
    \label{eq:affine-E1}\\
    \mathcal{E}_{n,2}
    &=
    \alpha_n\beta_n,
    \label{eq:affine-E2}\\
    \mathcal{E}_{n,3}
    &=
    \frac{\alpha_n\beta_n}{3}
    \left(
        \alpha_n^2+3\alpha_n\beta_n+\beta_n^2
    \right),
    \label{eq:affine-E3}\\
    \mathcal E_{n,4}
    &=
    \frac{-\alpha_n\beta_n(\beta_n-\alpha_n)}{45}
    \left(
        7\alpha_n^2+20\alpha_n\beta_n+7\beta_n^2
    \right),
    \label{eq:affine-E4}\\
    \mathcal{E}_{n,5}
    &=
    \frac{(\alpha_{n}+\beta_{n})^{2}}{5} \mathcal{E}_{n,3} + \frac{\alpha_{n}^{3}\beta_{n}^{3}}{45}.
    \label{eq:affine-E5}
\end{align}
\end{subequations}
Moreover, the optimal coefficients are given by
\begin{equation}
    c_n^{*}
    =
    \frac{
        \mathcal{E}_{n,3}\mathcal{E}_{n,4}
        -2\mathcal{E}_{n,5}\mathcal{E}_{n,1}
    }{
        4\mathcal{E}_{n,2}\mathcal{E}_{n,5}
        -\mathcal{E}_{n,3}^2
    },\,
    b_n^{*}
    =
    \frac{
        \mathcal{E}_{n,3}\mathcal{E}_{n,1}
        -2\mathcal{E}_{n,2}\mathcal{E}_{n,4}
    }{
        4\mathcal{E}_{n,2}\mathcal{E}_{n,5}
        -\mathcal{E}_{n,3}^2
    }.
    \label{eq:oac-coeffs}
\end{equation}
\end{theorem}
\begin{IEEEproof}
By the strong Markov property, the shifted process
$\widetilde{W}_u=W_{S_n+u}-W_{S_n}$ is a standard Wiener process, and $T_n$ is its first exit time from $(-\alpha_n,\beta_n)$. Then, we may write
\begin{align*}
    \mathsf{MSE}^{\mathrm{AFF}}_n(c_n, b_n)
    &=
    \mathbb{E}\left[
    \int_{S_n}^{S_{n+1}}
    (W_t-\hat{W}_t)^2\,dt
    \right]\\
    &=
    \mathbb{E}\left[
    \int_0^{T_n}
    \left(\widetilde{W}_u - c_n - b_n u \right)^2\,du
    \right].
\end{align*}
For notational simplicity, we drop the subscript $n$ and write 
\begin{equation*}
    W=\widetilde{W},\,
    \tau=T_n,\,
    \alpha=\alpha_n,\,
    \beta=\beta_n,\,
    c=c_n,\, b=b_n.
\end{equation*}
Expanding the squared error gives
\begin{align*}
    &\mathbb{E}\left[
        \int_{0}^{\tau}(W_{u}-c-bu)^{2}\,du
    \right]
    \notag\\
    =&\,
    \mathbb{E}\left[\int_{0}^{\tau}W_{u}^{2}\,du\right]
    -2c\mathbb{E}\left[\int_{0}^{\tau}W_{u}\,du\right]
    +c^{2}\mathbb{E}\left[\int_{0}^{\tau}1\,du\right]
    \notag\\
    +& 2cb\mathbb{E}\left[\int_{0}^{\tau}u\,du\right]
    -2b\mathbb{E}\left[\int_{0}^{\tau}uW_{u}\,du\right]
    +b^{2}\mathbb{E}\left[\int_{0}^{\tau}u^{2}\,du\right].
\end{align*}
We therefore evaluate the following quantities
\begin{equation*}
\begin{alignedat}{2}
    \mathcal{E}_{0}
    &= \mathbb{E}\left[\int_{0}^{\tau}W_{u}^{2}\,du\right],
    &\quad \mathcal{E}_{1}
    &= -2\mathbb{E}\left[\int_{0}^{\tau}W_{u}\,du\right],\\
    \mathcal{E}_{2}
    &= \mathbb{E}\left[\int_{0}^{\tau}1\,du\right],
    &\quad \mathcal{E}_{3}
    &= 2\mathbb{E}\left[\int_{0}^{\tau}u\,du\right],\\
    \mathcal{E}_{4}
    &= -2\mathbb{E}\left[\int_{0}^{\tau}uW_{u}\,du\right],
    &\quad \mathcal{E}_{5}
    &= \mathbb{E}\left[\int_{0}^{\tau}u^{2}\,du\right].
\end{alignedat}
\end{equation*}

By continuity, $W_{\tau}\in\{-\alpha,\beta\}$ almost surely, and the first-exit probabilities are
\begin{equation*}
    \mathbf{P}(W_{\tau}=\beta)
    =
    \frac{\alpha}{\alpha+\beta},
    \quad
    \mathbf{P}(W_{\tau}=-\alpha)
    =
    \frac{\beta}{\alpha+\beta}.
\end{equation*}
Consequently, for each integer $m\geq1$,
\begin{equation}
    \mathbb{E}[W_{\tau}^{m}]
    =
    \frac{\alpha\beta^{m}+\beta(-\alpha)^{m}}
    {\alpha+\beta}.
    \label{eq:affine-exit-state-moments}
\end{equation}
Applying the stopped polynomial-martingale identity~\eqref{eq:app-poly-martingale} from Appendix~\ref{app:wiener-martingales} yields
\begin{align}
    \mathbb{E}\left[\int_{0}^{\tau}1\,du\right]
    &=
    \mathbb{E}[W_{\tau}^{2}]
    =
    \alpha\beta,
    \label{eq:affine-first-exit-moment}\\
    \mathbb{E}\left[\int_{0}^{\tau}W_{u}\,du\right]
    &=
    \frac{1}{3}\mathbb{E}[W_{\tau}^{3}]
    =
    \frac{\alpha\beta(\beta-\alpha)}{3},
    \label{eq:affine-state-integral}\\
    \mathbb{E}\left[\int_{0}^{\tau}W_{u}^{2}\,du\right]
    &=
    \frac{1}{6}\mathbb{E}[W_{\tau}^{4}]
    =
    \frac{\alpha\beta}{6}
    (\alpha^{2}-\alpha\beta+\beta^{2}).
    \label{eq:affine-square-integral}
\end{align}
These identities yield $\mathcal{E}_{0}$, $\mathcal{E}_{1}$, and $\mathcal{E}_{2}$.

To evaluate the integrals defining $\mathcal{E}_{3}$ and $\mathcal{E}_{5}$, expand the first-exit Laplace transform~\eqref{eq:LT_Tn} using
\begin{equation*}
    \cosh(x\sqrt{2s})
    =
    1+x^{2}s+\frac{x^{4}}{6}s^{2}
    +\frac{x^{6}}{90}s^{3}
    +\mathcal{O}(s^{4}).
\end{equation*}
This gives
\begin{align*}
    \mathcal{L}_{\tau}(s)
    &=
    1-\alpha\beta s
    +\frac{\alpha\beta}{6}
    (\alpha^{2}+3\alpha\beta+\beta^{2})s^{2}
    -\frac{\alpha\beta}{90}
    \Big(
        3\alpha^{4}\\
    &\quad
        +15\alpha^{3}\beta
        +25\alpha^{2}\beta^{2}
        +15\alpha\beta^{3}+3\beta^{4}
    \Big)s^{3}
    +\mathcal{O}(s^{4}).
\end{align*}
For every nonnegative integer $j$,
\begin{equation*}
    \mathbb{E}\left[\int_{0}^{\tau}u^{j}\,du\right]
    =
    \frac{(-1)^{j+1}}{j+1}
    \mathcal{L}_{\tau}^{(j+1)}(0).
\end{equation*}
Therefore, we obtain
\begin{equation}
    \mathbb{E}\left[\int_{0}^{\tau}u\,du\right]
    =
    \frac{\alpha\beta}{6}
    (\alpha^{2}+3\alpha\beta+\beta^{2}),
    \label{eq:affine-second-exit-moment}
\end{equation}
and
\begin{align}
    &\mathbb{E}\left[\int_{0}^{\tau}u^{2}\,du\right] \notag\\
    &=
    \frac{\alpha\beta}{45}
    \Big(
        3\alpha^{4}+15\alpha^{3}\beta
        +25\alpha^{2}\beta^{2}
        +15\alpha\beta^{3}+3\beta^{4}
    \Big) \notag\\
    &= \frac{\alpha_{n}\beta_{n}}{45}\Big(3(\alpha_{n}+\beta_{n})^{2} (\alpha^{2}+3\alpha\beta+\beta^{2})
    + \alpha_{n}^{2}\beta_{n}^{2}\Big),
    \label{eq:affine-third-exit-moment}
\end{align}
which give $\mathcal{E}_{3}$ and $\mathcal{E}_{5}$.

% For the remaining exit-time moments, expand the Laplace transform~\eqref{eq:LT_Tn} using 
% \begin{equation*}
%     \cosh(x\sqrt{2s})
%     =
%     1+x^2s+\frac{x^4}{6}s^2
%     +\frac{x^6}{90}s^3+ \mathcal{O}(s^4).
% \end{equation*}
% This gives
% \begin{align*}
%     \mathcal{L}_\tau(s)
%     &=
%     1-\alpha\beta s
%     +\frac{\alpha\beta}{6}
%     (\alpha^2+3\alpha\beta+\beta^2)s^2 
%     -\frac{\alpha\beta}{90}
%     \Big(
%     3\alpha^4 \\
%     &\quad +15\alpha^3\beta+25\alpha^2\beta^2
%         +15\alpha\beta^3+3\beta^4
%     \Big)s^3
%     +\mathcal{O}(s^4).
% \end{align*}
% Since $\mathcal{L}_\tau^{(k)}(0)=(-1)^k\mathbb{E}[\tau^k]$, we obtain
% \begin{align}
%     \mathbb{E}[\tau^2]
%     &=
%     \frac{\alpha\beta}{3}
%     (\alpha^2+3\alpha\beta+\beta^2),
%     \label{eq:affine-second-exit-moment}\\
%     \mathbb{E}[\tau^3]
%     &=
%     \frac{\alpha\beta}{15}
%     \left(
%         3\alpha^4+15\alpha^3\beta+25\alpha^2\beta^2
%         +15\alpha\beta^3+3\beta^4
%     \right),
%     \label{eq:affine-third-exit-moment}
% \end{align}
% which yields $\mathcal{E}_{3}$ and $\mathcal{E}_{5}$.

It remains to evaluate the mixed integral $\mathbb{E}[\int_0^\tau uW_u\,du]$. By It\^o's formula (see Appendix~\ref{app:wiener-martingales}), 
\begin{equation*}
    d(uW_{u}^{3})
    =
    (W_{u}^{3}+3uW_{u})\,du
    +3uW_{u}^{2}\,dW_{u}.
\end{equation*}
Integrating up to $\tau$ gives
\begin{equation*}
    \tau W_{\tau}^{3}
    =
    \int_{0}^{\tau}W_{u}^{3}\,du
    +3\int_{0}^{\tau}uW_{u}\,du
    +3\int_{0}^{\tau}uW_{u}^{2}\,dW_{u}.
\end{equation*}
The last term on the RHS is a stochastic integral. 
By Lemma~\ref{lem:app-stopped-stochastic-integrals}, it is
a martingale with zero expectation. Hence, taking expectations on both sides yields
\begin{equation}
    3\mathbb{E}\left[\int_{0}^{\tau}uW_{u}\,du\right]
    =
    \mathbb{E}[\tau W_{\tau}^{3}]
    -
    \mathbb{E}\left[\int_{0}^{\tau}W_{u}^{3}\,du\right].
    \label{eq:affine-mixed-integral-identity}
\end{equation}
We next evaluate the two terms on the RHS.

For the first term, It\^o's rule gives
\begin{equation*}
    d(uW_{u})=W_{u}\,du+u\,dW_{u}.
\end{equation*}
Integrating up to $\tau$ and applying
Lemma~\ref{lem:app-stopped-stochastic-integrals}
with $g(u)=u$ and $m=0$ yields
\begin{equation}
    \mathbb{E}[\tau W_{\tau}]
    =
    \mathbb{E}\left[\int_{0}^{\tau}W_{u}\,du\right]
    =
    \frac{\alpha\beta(\beta-\alpha)}{3}.
    \label{eq:affine-time-state}
\end{equation}
Since $W_{\tau}\in\{-\alpha,\beta\}$, we have
\begin{equation*}
    (W_{\tau}+\alpha)(W_{\tau}-\beta)=0,
\end{equation*}
and therefore
\begin{equation*}
    W_{\tau}^{2}
    =
    (\beta-\alpha)W_{\tau}+\alpha\beta.
\end{equation*}
Multiplying by $W_{\tau}$ and substituting this identity again gives
\begin{align*}
    W_{\tau}^{3}
    &=
    (\beta-\alpha)W_{\tau}^{2}+\alpha\beta W_{\tau}\\
    &=
    (\alpha^{2}-\alpha\beta+\beta^{2})W_{\tau}
    +\alpha\beta(\beta-\alpha).
\end{align*}
Multiplying by $\tau$ and using~\eqref{eq:affine-first-exit-moment} and~\eqref{eq:affine-time-state}, we obtain
\begin{align}
    \mathbb{E}[\tau W_{\tau}^{3}]
    &=
    (\alpha^{2}-\alpha\beta+\beta^{2})
    \mathbb{E}[\tau W_{\tau}]
    \notag\\
    &\quad+
    \alpha\beta(\beta-\alpha)
    \mathbb{E}\left[\int_{0}^{\tau}1\,du\right]
    \notag\\
    &=
    \frac{\alpha\beta(\beta-\alpha)(\alpha+\beta)^{2}}{3}.
    \label{eq:affine-time-cube}
\end{align}

For the second term, the stopped polynomial-martingale identity~\eqref{eq:app-poly-martingale} with $m=5$, together with~\eqref{eq:affine-exit-state-moments}, gives
\begin{align}
    \mathbb{E}\left[\int_{0}^{\tau}W_{u}^{3}\,du\right]
    &=
    \frac{1}{10}\mathbb{E}[W_{\tau}^{5}]
    =
    \frac{\alpha\beta(\beta-\alpha)
    (\alpha^{2}+\beta^{2})}{10}.
    \label{eq:affine-cube-integral}
\end{align}
Substituting~\eqref{eq:affine-time-cube} and~\eqref{eq:affine-cube-integral} into~\eqref{eq:affine-mixed-integral-identity} gives
\begin{align}
    \mathbb{E}\left[\int_{0}^{\tau}uW_{u}\,du\right]
    &=
    \frac{\alpha\beta(\beta-\alpha)}{90}
    \left(
        10(\alpha+\beta)^{2}
        -3(\alpha^{2}+\beta^{2})
    \right)
    \notag\\
    &=
    \frac{\alpha\beta(\beta-\alpha)}{90}
    \left(
        7\alpha^{2}+20\alpha\beta+7\beta^{2}
    \right).
    \label{eq:affine-age-state-integral}
\end{align}
This yields $\mathcal{E}_{4}$ and completes the proof
of~\eqref{eq:affine-age-mse}.

To determine the optimal coefficients, 
we first minimize the quadratic
in~\eqref{eq:affine-age-mse} over $c$ for a fixed $b$. Since $\mathcal{E}_{2}>0$, we have
\begin{equation*}
    c(b)
    =
    -\frac{
        \mathcal{E}_{1}+\mathcal{E}_{3}b
    }{2\mathcal{E}_{2}}.
\end{equation*}
Substituting this expression into~\eqref{eq:affine-age-mse} gives
\begin{align*}
    \mathsf{MSE}^{\mathrm{AFF}}(c, b) &= \mathcal{E}_{0}
    -\frac{\mathcal{E}_{1}^2}{4\mathcal{E}_{2}}
    +
    \big(
        \mathcal{E}_{4}
        -
        \frac{\mathcal{E}_{1}\mathcal{E}_{3}}
        {2\mathcal{E}_{2}}
    \big)b
    +
    \big(
        \mathcal{E}_{5}
        -
        \frac{\mathcal{E}_{3}^2}
        {4\mathcal{E}_{2}}
    \big)b^2.
\end{align*}
It can be shown that the coefficient of $b^2$ is positive. Minimizing this quadratic yields
\begin{equation*}
    b^*
    =
    \frac{
        \mathcal{E}_{3}\mathcal{E}_{1}
        -2\mathcal{E}_{2}\mathcal{E}_{4}
    }{
        4\mathcal{E}_{2}\mathcal{E}_{5}
        -\mathcal{E}_{3}^2
    }, \qquad c^* = c(b^*).
\end{equation*}
This completes the proof.
\end{IEEEproof}

We next discuss three specific estimators within this class.

\begin{enumerate}
    \item[(1)] \emph{ZOH estimator:} No correction is applied to the most recently received
    sample, i.e.,
    \begin{equation}
        c_n^{\mathrm{ZOH}}=0,
        \qquad
        b_n^{\mathrm{ZOH}}=0.
        \label{eq:zoh-coefficients}
    \end{equation}
    The ZOH estimator is MMSE-optimal under deviation-based sampling and is widely used for Wiener processes.
    \item[(2)] \emph{LSI estimator:}
    This estimator applies the long-silence limit correction in Lemma~\ref{lemma:LS}, i.e.,
    \begin{equation}
        c_n^{\mathrm{LSI}}
        =
        \lim_{u\to\infty}f_n(u)
        =
        \frac{\beta_n-\alpha_n}{2}, \quad b_n^{\mathrm{LSI}}=0.
        \label{eq:ls-correction}
    \end{equation}
    % \item[(3)] \emph{OCC estimator:} The OCC estimator is a constant-correction estimator where $b_n = 0$ and $c_n$ is chosen to minimize the MSE~\eqref{eq:affine-age-mse} over the $n$th interval, i.e.,
    % \begin{equation}
    %     c_n^{\mathrm{OCC}}
    %     =
    %     \frac{
    %     \mathbb{E}[
    %     \int_0^{T_n}\widetilde W_t \,dt]
    %     }{
    %     \mathbb{E}[T_n]
    %     }
    %     =
    %     \frac{\beta_n-\alpha_n}{3}, \,\, b_n^{\mathrm{OCC}} = 0.
    %     \label{eq:occ-correction}
    % \end{equation}
    \item[(3)] \emph{OAC estimator:}
    This estimator uses the optimal coefficients in~\eqref{eq:oac-coeffs}, i.e.,
    \begin{align}
        c_n^{\mathrm{OAC}} = c_n^*, \qquad
        b_n^{\mathrm{OAC}} = b_n^*.
    \end{align}
\end{enumerate}

% \begin{figure}[t!]
%     \centering
%     \includegraphics[width=\linewidth]{images/correction.pdf}
%     \caption{Comparison of correction terms for $\delta=1$, $D_{n}=-1$, and $n = 10$.}
%  \label{fig:correction}
% \end{figure}

% \begin{figure}[t!]
%     \centering
%     \includegraphics[width=\linewidth]{images/exp_correction.pdf}
%     \caption{Convergence of the EXP estimator for $\delta=1$, $D_{n}=-1$, and $n=10$.}
%     \label{fig:exp-correction}
% \end{figure}

\begin{figure}[t!]
    \centering
    \subfloat[affine-age estimators]{
    \includegraphics[width=0.97\linewidth]{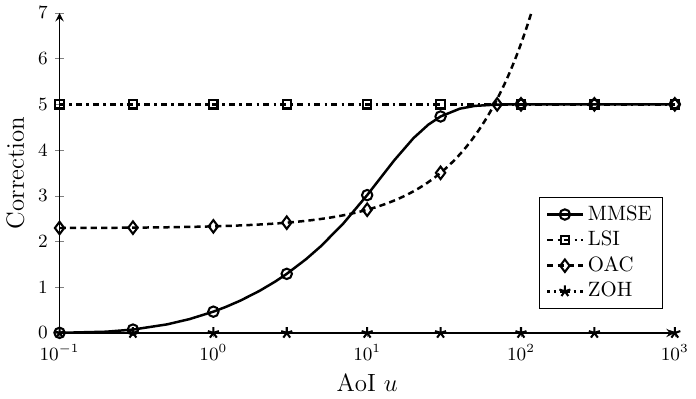}
    \label{fig:correction}
    }
    \\
    \subfloat[exponential-age estimator]{
    \includegraphics[width=0.97\linewidth]{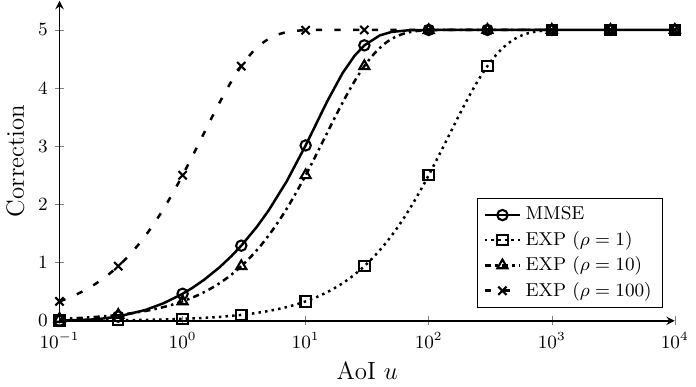}
    \label{fig:exp-correction}
    }
    \\
    \caption{Comparison of correction terms for $\delta=1$, $D_{n}=-1$, and $n = 10$.}
    \label{fig:comparison-correction}
\end{figure}

Fig.~\ref{fig:correction} illustrates how the estimators respond to the absence of a new innovation. Immediately after an innovation, the MMSE estimate coincides with the ZOH estimate. As the AoI increases, the MMSE correction converges exponentially to the LSI correction. In contrast, the OAC correction varies linearly with AoI and is unbounded whenever its slope is nonzero.

The next lemma characterizes the estimation performance of these estimators.

\begin{lemma}
\label{lemma:affine-age-average-mse}
Under innovation-based sampling, the average MSEs of the ZOH, LSI, and OAC estimators over the first $n$ sampling intervals, for $n\geq 1$, are
\begin{subequations}
\label{eq:affine-average-mses}
\begin{align}
    \bar{\mathcal{J}}_{n}^{\mathrm{ZOH}}
    &=
    \frac{3n^{2}-n+4}{36}\delta^{2},
    \label{eq:zoh-average-mse}\\
    \bar{\mathcal{J}}_{n}^{\mathrm{LSI}}
    &=
    \frac{n^{2}+n+2}{24}\delta^{2},
    \label{eq:ls-average-mse}\\
    \bar{\mathcal{J}}_{n}^{\mathrm{OAC}}
    &=
    \frac{3n^{2}+7n+8}{108}\delta^{2}
    -
    \frac{2\delta^{2}}{45n(n+1)} \notag\\
    &\quad
    \times \sum_{k=0}^{n-1}
    \frac{k^{2}(k+1)(2k^{2}+9k+9)^{2}}
    {7k^{4}+58k^{3}+177k^{2}+238k+119}.
    \label{eq:oac-average-mse}
\end{align}
\end{subequations}
Moreover, their relative gains over ZOH satisfy
\begin{subequations}
\label{eq:affine-asymptotic-gains}
\begin{align}
    \frac{
        \bar{\mathcal{J}}_{n}^{\mathrm{ZOH}}
        -\bar{\mathcal{J}}_{n}^{\mathrm{LSI}}
    }{
        \bar{\mathcal{J}}_{n}^{\mathrm{ZOH}}
    }
    &\longrightarrow \frac{1}{2},
    \label{eq:lsi-asymptotic-gain}\\
    \frac{
        \bar{\mathcal{J}}_{n}^{\mathrm{ZOH}}
        -\bar{\mathcal{J}}_{n}^{\mathrm{OAC}}
    }{
        \bar{\mathcal{J}}_{n}^{\mathrm{ZOH}}
    }
    &\longrightarrow \frac{26}{35},
    \label{eq:oac-asymptotic-gain}
\end{align}
\end{subequations}
as $n\to\infty$.
\end{lemma}

\begin{IEEEproof}
We apply Theorem~\ref{thm:affine-age-estimator} to each coefficient choice.

For the ZOH estimator,
$c_{k}^{\mathrm{ZOH}}=b_{k}^{\mathrm{ZOH}}=0$
during the $k$th sampling interval. Substituting these coefficients into~\eqref{eq:affine-age-mse} gives
\begin{equation*}
    \mathsf{MSE}_{k}^{\mathrm{ZOH}}
    =
    \frac{\alpha_{k}\beta_{k}}{6}
    (\alpha_{k}^{2}-\alpha_{k}\beta_{k}+\beta_{k}^{2}).
\end{equation*}
Using $\{\alpha_{k},\beta_{k}\}=\{\delta,(k+1)\delta\}$ yields
\begin{equation*}
    \mathsf{MSE}_{k}^{\mathrm{ZOH}}
    =
    \frac{(k+1)(k^{2}+k+1)}{6}\delta^{4}.
\end{equation*}
Summing over the first $n$ sampling intervals yields
\begin{align}
    \mathcal{J}_{n}^{\mathrm{ZOH}}
    &=
    \frac{\delta^{4}}{6}
    \sum_{k=0}^{n-1}(k+1)(k^{2}+k+1)
    \notag\\
    &=
    \frac{n(n+1)(3n^{2}-n+4)}{72}\delta^{4}.
    \label{eq:ZOH-total-MSE}
\end{align}

For the LSI estimator,
$c_{k}^{\mathrm{LSI}}=(\beta_{k}-\alpha_{k})/2$
and $b_{k}^{\mathrm{LSI}}=0$. Substitution
into~\eqref{eq:affine-age-mse} gives
\begin{align*}
    \mathsf{MSE}_{k}^{\mathrm{LSI}}
    &=
    \frac{\alpha_{k}\beta_{k}}{6}
    (\alpha_{k}^{2}-\alpha_{k}\beta_{k}+\beta_{k}^{2})\\
    &\quad
    -\frac{\alpha_{k}\beta_{k}}{3}
    (\beta_{k}-\alpha_{k})^{2}
    +\frac{\alpha_{k}\beta_{k}}{4}
    (\beta_{k}-\alpha_{k})^{2}\\
    &=
    \frac{\alpha_{k}\beta_{k}}{12}
    (\alpha_{k}^{2}+\beta_{k}^{2})\\
    &=
    \frac{(k+1)(k^{2}+2k+2)}{12}\delta^{4}.
\end{align*}
Consequently,
\begin{align}
    \mathcal{J}_{n}^{\mathrm{LSI}}
    &=
    \frac{\delta^{4}}{12}
    \sum_{k=0}^{n-1}(k+1)(k^{2}+2k+2)
    \notag\\
    &=
    \frac{n(n+1)(n^{2}+n+2)}{48}\delta^{4}.
    \label{eq:LS-total-MSE}
\end{align}

For the OAC estimator, 
Substituting~\eqref{eq:oac-coeffs} into~\eqref{eq:affine-age-mse} gives
\begin{align*}
    \mathsf{MSE}_{k}^{\mathrm{OAC}}
    &=
    \mathcal{E}_{k,0}
    -
    \frac{\mathcal{E}_{k,1}^{2}}{4\mathcal{E}_{k,2}}
    -
    \frac{
        \left(
            \mathcal{E}_{k,3}\mathcal{E}_{k,1}
            -2\mathcal{E}_{k,2}\mathcal{E}_{k,4}
        \right)^{2}
    }{
        4\mathcal{E}_{k,2}
        \left(
            4\mathcal{E}_{k,2}\mathcal{E}_{k,5}
            -\mathcal{E}_{k,3}^{2}
        \right)
    }.
\end{align*}
Using $\beta_{k}-\alpha_{k}=-D_{k}k\delta$
and $\alpha_{k}\beta_{k}=(k+1)\delta^{2}$, we obtain
\begin{align*}
    \mathsf{MSE}_{k}^{\mathrm{OAC}}
    &=
    \frac{(k+1)(k^{2}+3k+3)}{18}\delta^{4}\\
    &\quad
    -
    \frac{
        k^{2}(k+1)(2k^{2}+9k+9)^{2}
    }{
        45(7k^{4}+58k^{3}+177k^{2}+238k+119)
    }\delta^{4}.
\end{align*}
Summing over the first $n$ sampling intervals gives
\begin{align}
    \mathcal{J}_{n}^{\mathrm{OAC}}
    &=
    \frac{n(n+1)(3n^{2}+7n+8)}{216}\delta^{4}
    \notag\\
    &
    -
    \frac{\delta^{4}}{45}
    \sum_{k=0}^{n-1}
    \frac{
        k^{2}(k+1)(2k^{2}+9k+9)^{2}
    }{
        7k^{4}+58k^{3}+177k^{2}+238k+119
    }.
    \label{eq:OAC-total-MSE}
\end{align}

By Theorem~\ref{thm:sampling_times},
$\mathbb{E}[S_{n}]=n(n+1)\delta^{2}/2$.
Dividing~\eqref{eq:ZOH-total-MSE},
\eqref{eq:LS-total-MSE}, and~\eqref{eq:OAC-total-MSE}
by $\mathbb{E}[S_{n}]$ proves~\eqref{eq:affine-average-mses}.

We next establish the asymptotic gains.
For LSI,
\begin{equation*}
    \frac{
        \bar{\mathcal{J}}_{n}^{\mathrm{ZOH}}
        -\bar{\mathcal{J}}_{n}^{\mathrm{LSI}}
    }{
        \bar{\mathcal{J}}_{n}^{\mathrm{ZOH}}
    }
    =
    1-
    \frac{3(n^{2}+n+2)}
    {2(3n^{2}-n+4)}
    \longrightarrow
    \frac{1}{2}.
\end{equation*}
For OAC, the per-interval MSE satisfies
\begin{align*}
    \mathsf{MSE}_{k}^{\mathrm{OAC}}
    &=
    \left(\frac{1}{18}-\frac{4}{315}\right)
    k^{3}\delta^{4}
    +\mathcal{O}(k^{2}\delta^{4})\\
    &=
    \frac{3}{70}k^{3}\delta^{4}
    +\mathcal{O}(k^{2}\delta^{4}).
\end{align*}
Using
$\sum_{k=0}^{n-1}k^{3}=n^{2}(n-1)^{2}/4$
and $\sum_{k=0}^{n-1}k^{2}=\mathcal{O}(n^{3})$ gives
\begin{equation*}
    \mathcal{J}_{n}^{\mathrm{OAC}}
    =
    \frac{3}{280}n^{4}\delta^{4}
    +\mathcal{O}(n^{3}\delta^{4}).
\end{equation*}
Consequently,
\begin{equation*}
    \bar{\mathcal{J}}_{n}^{\mathrm{OAC}}
    =
    \frac{3}{140}n^{2}\delta^{2}
    +\mathcal{O}(n\delta^{2}).
\end{equation*}
Thus,
\begin{equation*}
    \frac{
        \bar{\mathcal{J}}_{n}^{\mathrm{ZOH}}
        -\bar{\mathcal{J}}_{n}^{\mathrm{OAC}}
    }{
        \bar{\mathcal{J}}_{n}^{\mathrm{ZOH}}
    }
    \longrightarrow
    1-\frac{3/140}{1/12}
    =
    \frac{26}{35}.
\end{equation*}
This completes the proof.
\end{IEEEproof}

% \begin{figure}[t!]
%     \centering
%     \includegraphics[width=\linewidth]{images/asymptotic_gain.pdf}
%     \caption{Performance gains of the proposed estimators relative to ZOH.}
%     \label{fig:asymptotic-gain}
% \end{figure}

% \begin{figure}
%     \centering
%     \includegraphics[width=\linewidth]{images/exp_gain.pdf}
%     \caption{Performance gain of the exponential-age estimator.}
%     \label{fig:exp_gain}
% \end{figure}

\begin{figure}[t!]
    \centering
    \subfloat[affine-age estimators]{
    \includegraphics[width=0.97\linewidth]{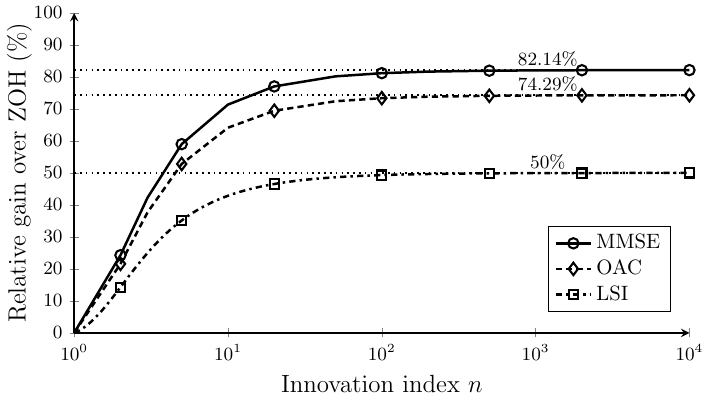}
    \label{fig:asymptotic-gain}
    }
    \\
    \subfloat[exponential-age estimator]{
    \includegraphics[width=0.97\linewidth]{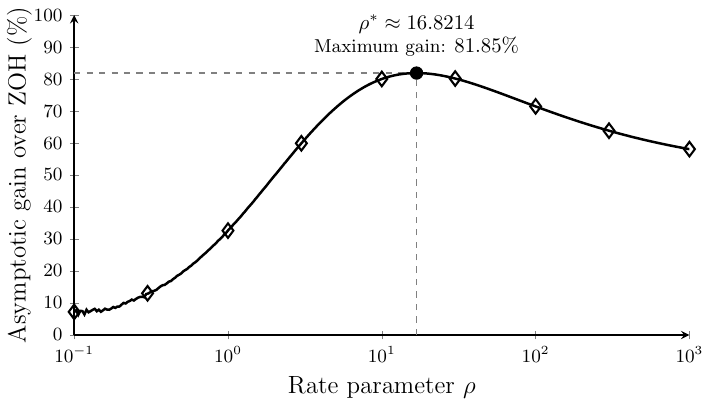}
    \label{fig:exp_gain}
    }
    \\
    \caption{Performance gains of the proposed estimators relative to ZOH.}
    \label{fig:comparison-gains}
\end{figure}

The relative performance gains of these estimators over ZOH are reported in Fig.~\ref{fig:asymptotic-gain}. The LSI and OAC gains converge to $50\%$ and $74.29\%$, respectively. The MMSE gain is evaluated numerically using Lemma~\ref{lem:mmse-mse} and approaches $82.14\%$. \emph{This highlights the importance of exploiting the information conveyed by silence.}

\subsection{Exponential-Age Estimators}
\label{sec:exp-age}

The affine-age estimators, while analytically simple, may lead to unbounded corrections. We therefore introduce an exponential-age estimator whose correction has a finite limit.

\begin{definition}
An \emph{exponential-age estimator} has the form
\begin{equation}
    \hat{W}_{t}
    =
    W_{S_{n}}+c_{n}(1-e^{-r_{n}U_{t}}),
    \quad t\in[S_{n},S_{n+1}),
    \label{eq:exponential-age-estimator}
\end{equation}
where
\begin{equation}
    c_{n}
    =
    c_{n}^{\mathrm{LSI}}
    =
    \frac{\beta_{n}-\alpha_{n}}{2},
    \qquad
    r_{n}
    =
    \frac{\rho}{(\alpha_{n}+\beta_{n})^{2}},\label{eq:exp-coeff}
\end{equation}
and $\rho>0$ is a rate parameter. 
\end{definition}

As illustrated in~Fig.~\ref{fig:exp-correction}, the EXP correction starts at zero and approaches the MMSE long-silence limit as AoI goes to infinity. The parameter $\rho$ controls how quickly the correction approaches this limit.

The next result characterizes the estimation performance and asymptotic gain of this estimator in closed form.

\begin{theorem}
\label{thm:exponential-age-mse}
The MSE of the exponential-age estimator~\eqref{eq:exponential-age-estimator} over the $n$th sampling interval is
\begin{equation}
    \mathsf{MSE}_{n}^{\mathrm{EXP}}(\rho)
    =
    \mathcal{E}_{n,0}
    -2A_{n}(r_{n})c_{n}
    +B_{n}(r_{n})c_{n}^{2},
    \label{eq:exponential-age-mse}
\end{equation}
where $\mathcal{E}_{n,0}$ is given
in~\eqref{eq:affine-E0}, and, for $r>0$,
\begin{subequations}
\begin{align}
    A_{n}(r)
    &=
    \frac{\alpha_{n}\beta_{n}(\beta_{n}-\alpha_{n})}{3}
    \notag\\
    &\quad-
    \frac{
        \alpha_{n}\sinh(\beta_{n}\sqrt{2r})
        -\beta_{n}\sinh(\alpha_{n}\sqrt{2r})
    }{
        r\sinh((\alpha_{n}+\beta_{n})\sqrt{2r})
    },
    \label{eq:exponential-A}\\
    B_{n}(r)
    &=
    \alpha_{n}\beta_{n}
    +
    \frac{
        4\mathcal{L}_{T_{n}}(r)
        -\mathcal{L}_{T_{n}}(2r)-3
    }{2r},
    \label{eq:exponential-B}
\end{align}
\end{subequations}
with $\mathcal{L}_{T_{n}}$ given in~\eqref{eq:LT_Tn}.
Moreover, as $n\to\infty$,
\begin{equation}
    \bar{\mathcal{J}}_{n}^{\mathrm{EXP}}(\rho)
    =
    \frac{h(\rho)}{2}n^{2}\delta^{2}
    +\mathcal{O}(n\delta^{2}),
    \label{eq:exp-average-mse-asymptotic}
\end{equation}
where
\begin{equation}
    h(\rho)
    =
    \frac{1}{12}
    +\frac{1}{\rho}
    -\frac{\coth(\sqrt{\rho/2})}{\sqrt{2\rho}}
    +\frac{\tanh(\sqrt{\rho})}{4\sqrt{\rho}},
    \label{eq:hrho}
\end{equation}
and the asymptotic gain over ZOH is
\begin{equation}
    \lim_{n\to\infty}
    \frac{
        \bar{\mathcal{J}}_{n}^{\mathrm{ZOH}}
        -\bar{\mathcal{J}}_{n}^{\mathrm{EXP}}(\rho)
    }{
        \bar{\mathcal{J}}_{n}^{\mathrm{ZOH}}
    }
    =
    1-6h(\rho).
\end{equation}
\end{theorem}

\begin{IEEEproof}
Fix the sampling history $\mathcal{H}_{n}$ and suppress the subscript $n$. As in the proof of Theorem~\ref{thm:affine-age-estimator}, let $W_{u}$
denote the shifted Wiener process and let $\tau$ be its first exit time from $(-\alpha,\beta)$. All expectations below are conditional on this history. The conditional MSE is
\begin{align*}
    \mathsf{MSE}^{\mathrm{EXP}}(c,r)
    &=
    \mathbb{E}\left[
        \int_{0}^{\tau}
        \big(W_{u}-c(1-e^{-ru})\big)^{2}\,du
    \right]\\
    &=
    \mathcal{E}_{0}-2cA(r)+c^{2}B(r),
\end{align*}
where
\begin{align*}
    A(r)
    &=
    \mathbb{E}\left[\int_{0}^{\tau}W_{u}\,du\right]
    -
    \mathbb{E}\left[\int_{0}^{\tau}e^{-ru}W_{u}\,du\right],\\
    B(r)
    &=
    \mathbb{E}\left[
        \int_{0}^{\tau}(1-e^{-ru})^{2}\,du
    \right].
\end{align*}

For $A(r)$, the first term is given in~\eqref{eq:affine-state-integral}. To evaluate the second term, It\^o's formula gives
\begin{equation*}
    d(e^{-ru}W_{u})
    =
    e^{-ru}\,dW_{u}
    -re^{-ru}W_{u}\,du.
\end{equation*}
Integrating up to $\tau$ and using $W_{0}=0$ gives
\begin{equation*}
    e^{-r\tau}W_{\tau}
    =
    \int_{0}^{\tau}e^{-ru}\,dW_{u}
    -
    r\int_{0}^{\tau}e^{-ru}W_{u}\,du.
\end{equation*}
The first term on the RHS is a stochastic integral. Since 
\begin{equation*}
    \mathbb{E}\left[
        \int_{0}^{\tau}e^{-2ru}\,du
    \right]
    \leq
    \frac{1}{2r}<\infty,
\end{equation*}
it follows from Lemma~\ref{lem:app-stopped-stochastic-integrals} that this stochastic integral is a martingale and has zero expectation. Hence, taking expectations on both sides yields
\begin{equation}
    \mathbb{E}\left[
        \int_{0}^{\tau}e^{-ru}W_{u}\,du
    \right]
    =
    -\frac{\mathbb{E}[e^{-r\tau}W_{\tau}]}{r}.
    \label{eq:exponential-discounted-state}
\end{equation}

The Laplace transforms of the first exit time $\tau$, separated according to the exit boundary, are given by (see~\eqref{eq:app-exit-boundary-transforms})
\begin{align*}
    \mathbb{E}\left[
        e^{-r\tau}\mathbb{I}\{W_{\tau}=\beta\}
    \right]
    &=
    \frac{\sinh(\alpha\sqrt{2r})}
         {\sinh((\alpha+\beta)\sqrt{2r})},\\
    \mathbb{E}\left[
        e^{-r\tau}\mathbb{I}\{W_{\tau}=-\alpha\}
    \right]
    &=
    \frac{\sinh(\beta\sqrt{2r})}
         {\sinh((\alpha+\beta)\sqrt{2r})}.
\end{align*}
Since $W_{\tau}\in\{-\alpha,\beta\}$ almost surely,
\begin{equation*}
    \mathbb{E}[e^{-r\tau}W_{\tau}]
    =
    \frac{
        \beta\sinh(\alpha\sqrt{2r})
        -\alpha\sinh(\beta\sqrt{2r})
    }{
        \sinh((\alpha+\beta)\sqrt{2r})
    }.
\end{equation*}
Substituting this into~\eqref{eq:exponential-discounted-state} gives
\begin{equation*}
    \mathbb{E}\left[
        \int_{0}^{\tau}e^{-ru}W_{u}\,du
    \right]
    =
    \frac{
        \alpha\sinh(\beta\sqrt{2r})
        -\beta\sinh(\alpha\sqrt{2r})
    }{
        r\sinh((\alpha+\beta)\sqrt{2r})
    }.
\end{equation*}
Subtracting this expression from
\eqref{eq:affine-state-integral} proves~\eqref{eq:exponential-A}.

For $B(r)$, expanding the square gives
\begin{align*}
    B(r)
    &=
    \mathbb{E}[\tau]
    -2\mathbb{E}\left[
        \int_{0}^{\tau}e^{-ru}\,du
    \right]
    +
    \mathbb{E}\left[
        \int_{0}^{\tau}e^{-2ru}\,du
    \right].
\end{align*}
For any $s>0$, direct integration and
$\mathcal{L}_{\tau}(s)=\mathbb{E}[e^{-s\tau}]$ yield
\begin{equation*}
    \mathbb{E}\left[
        \int_{0}^{\tau}e^{-su}\,du
    \right]
    =
    \frac{1-\mathcal{L}_{\tau}(s)}{s}.
\end{equation*}
Applying this identity at $s=r$ and $s=2r$, together with $\mathbb{E}[\tau]=\alpha\beta$ from~\eqref{eq:affine-first-exit-moment}, gives
\begin{align*}
    B(r)
    &=
    \alpha\beta
    -\frac{2(1-\mathcal{L}_{\tau}(r))}{r}
    +\frac{1-\mathcal{L}_{\tau}(2r)}{2r}\\
    &=
    \alpha\beta
    +
    \frac{
        4\mathcal{L}_{\tau}(r)
        -\mathcal{L}_{\tau}(2r)-3
    }{2r}.
\end{align*}
This proves~\eqref{eq:exponential-B}.

We next derive the average-MSE asymptotic.
Use $k$ to index the sampling intervals and write
$q=\sqrt{2\rho}$ for brevity.
Substituting $c_{k}=(\beta_{k}-\alpha_{k})/2$
into~\eqref{eq:exponential-age-mse}, the undiscounted terms combine to give
\begin{align*}
    &\mathcal{E}_{k,0}
    -\frac{\alpha_{k}\beta_{k}
        (\beta_{k}-\alpha_{k})^{2}}{3}
    +\frac{\alpha_{k}\beta_{k}
        (\beta_{k}-\alpha_{k})^{2}}{4}\\
    &\qquad=
    \frac{\alpha_{k}\beta_{k}
        (\alpha_{k}^{2}+\beta_{k}^{2})}{12}
    =
    \mathsf{MSE}_{k}^{\mathrm{LSI}}.
\end{align*}
Using
$\{\alpha_{k},\beta_{k}\}=\{\delta,(k+1)\delta\}$
and $r_{k}=\rho/(\alpha_{k} + \beta_{k})^{2}$
therefore gives
\begin{align}
    \frac{\mathsf{MSE}_{k}^{\mathrm{EXP}}(\rho)}
         {\delta^{4}}
    &=
    \frac{(k+1)(k^{2}+2k+2)}{12}
    \notag\\
    &+
    \frac{k(k+2)^{2}}{\rho}
    \frac{
        \sinh\big(\frac{k+1}{k+2}q\big)
        -(k+1)\sinh\big(\frac{q}{k+2}\big)
    }{\sinh(q)}
    \notag\\
    &+
    \frac{k^{2}(k+2)^{2}}{8\rho}
    \left[
        4\mathcal{L}_{T_{k}}(r_{k})
        -\mathcal{L}_{T_{k}}(2r_{k})-3
    \right].
    \label{eq:exp-scaled-interval-mse}
\end{align}

We evaluate these three terms as $k\to\infty$, with $\rho>0$ fixed. The first term satisfies
\begin{equation*}
    \frac{(k+1)(k^{2}+2k+2)}{12}
    =
    \frac{k^{3}}{12}+\mathcal{O}(k^{2}).
\end{equation*}
For the second term, Taylor expansion gives
\begin{align*}
    \sinh\left(\frac{k+1}{k+2}q\right)
    &=
    \sinh(q)+\mathcal{O}(k^{-1}),\\
    (k+1)\sinh\left(\frac{q}{k+2}\right)
    &=
    q+\mathcal{O}(k^{-1}).
\end{align*}
Consequently,
\begin{equation*}
    \frac{
        \sinh\left(\frac{k+1}{k+2}q\right)
        -(k+1)\sinh\left(\frac{q}{k+2}\right)
    }{\sinh(q)}
    =
    1-\frac{q}{\sinh(q)}+\mathcal{O}(k^{-1}).
\end{equation*}

For the third term, the Laplace transform in~\eqref{eq:LT_Tn} yields
\begin{align*}
    \mathcal{L}_{T_{k}}(r_{k})
    &=
    \frac{
        \cosh\left(\frac{q}{2}-\frac{q}{k+2}\right)
    }{\cosh(q/2)}\\
    &=
    1-\frac{q}{k+2}\tanh(q/2)
    +\mathcal{O}(k^{-2}).
\end{align*}
Replacing $q$ by $\sqrt{2}q$ gives
\begin{equation*}
    \mathcal{L}_{T_{k}}(2r_{k})
    =
    1-\frac{\sqrt{2}q}{k+2}\tanh(q/\sqrt{2})
    +\mathcal{O}(k^{-2}).
\end{equation*}
Hence,
\begin{align*}
    &4\mathcal{L}_{T_{k}}(r_{k})
    -\mathcal{L}_{T_{k}}(2r_{k})-3\\
    &\quad=
    \frac{
        -4q\tanh(q/2)
        +\sqrt{2}q\tanh(q/\sqrt{2})
    }{k+2}
    +\mathcal{O}(k^{-2}).
\end{align*}

Substituting these expansions into~\eqref{eq:exp-scaled-interval-mse} gives
\begin{equation}
    \mathsf{MSE}_{k}^{\mathrm{EXP}}(\rho)
    =
    h(\rho)k^{3}\delta^{4}
    +\mathcal{O}(k^{2}\delta^{4}),
    \label{eq:exp-interval-mse-asymptotic}
\end{equation}
where
\begin{align*}
    h(\rho)
    &=
    \frac{1}{12}
    +\frac{1-q/\sinh(q)}{\rho}\\
    &\quad+
    \frac{
        -4q\tanh(q/2)
        +\sqrt{2}q\tanh(q/\sqrt{2})
    }{8\rho}.
\end{align*}
Using the identity
\begin{equation*}
    \frac{1}{\sinh(q)}
    +\frac{1}{2}\tanh(q/2)
    =
    \frac{1}{2}\coth(q/2)
\end{equation*}
and recalling $q=\sqrt{2\rho}$, we obtain
\begin{equation*}
    h(\rho)
    =
    \frac{1}{12}+\frac{1}{\rho}
    -\frac{\coth(\sqrt{\rho/2})}{\sqrt{2\rho}}
    +\frac{\tanh(\sqrt{\rho})}{4\sqrt{\rho}},
\end{equation*}
which proves~\eqref{eq:hrho}.

Summing~\eqref{eq:exp-interval-mse-asymptotic} over the first $n$ sampling intervals and using $\sum_{k=0}^{n-1}k^{3} =n^{4}/4+\mathcal{O}(n^{3})$ gives
\begin{equation*}
    \mathcal{J}_{n}^{\mathrm{EXP}}(\rho)
    =
    \sum_{k=0}^{n-1}\mathsf{MSE}_{k}^{\mathrm{EXP}}(\rho)
    =
    \frac{h(\rho)}{4}n^{4}\delta^{4}
    +\mathcal{O}(n^{3}\delta^{4}).
\end{equation*}
Dividing by $\mathbb{E}[S_{n}]=n(n+1)\delta^{2}/2$
therefore yields
\begin{equation*}
    \bar{\mathcal{J}}_{n}^{\mathrm{EXP}}(\rho)
    =
    \frac{h(\rho)}{2}n^{2}\delta^{2}
    +\mathcal{O}(n\delta^{2}).
\end{equation*}

Finally, the ZOH average MSE satisfies
\begin{equation*}
    \bar{\mathcal{J}}_{n}^{\mathrm{ZOH}}
    =
    \frac{3n^{2}-n+4}{36}\delta^{2}
    =
    \frac{n^{2}\delta^{2}}{12}
    +\mathcal{O}(n\delta^{2}).
\end{equation*}
Taking the ratio of the leading terms gives
\begin{equation*}
    \lim_{n\to\infty}
    \frac{
        \bar{\mathcal{J}}_{n}^{\mathrm{ZOH}}
        -\bar{\mathcal{J}}_{n}^{\mathrm{EXP}}(\rho)
    }{
        \bar{\mathcal{J}}_{n}^{\mathrm{ZOH}}
    }
    =
    1-6h(\rho),
\end{equation*}
which completes the proof.
\end{IEEEproof}

Fig.~\ref{fig:exp_gain} illustrates the asymptotic gain of the exponential-age estimator. A small $\rho$ corrects the estimator too slowly to fully exploit the information conveyed by silence, while a large $\rho$ moves the estimate toward the long-silence limit prematurely. Numerical optimization gives $\rho^*\approx16.82$, yielding a gain of approximately $81.85\%$ over ZOH. Thus, an appropriately chosen exponential correction captures much of the estimation benefit of silence without evaluating the infinite series required by the exact MMSE estimator.

\section{Conclusion}
This work introduced innovation-based sampling, which generates a sample only when the source provides new information relative to all previous samples. Unlike periodic and deviation-based rules, which reference the elapsed time or the last sample, the innovation criterion references the entire sampling history. This history dependence does not entail growing complexity: the innovation process lies on a lattice and admits a three-dimensional state representation. For a Wiener process, innovations become progressively sparser and the sampling rate vanishes asymptotically. Their directions become increasingly persistent, and we established limit theorems for the number of direction reversals. The consequence for estimation is that silence carries information, and the MMSE estimate evolves with the AoI. Finally, we introduced tractable affine- and exponential-age approximations that recover most of the MMSE gain. Overall, innovations convey information through their content, direction, and timing.

\appendix
This appendix collects the necessary background on Wiener processes and stochastic integrals used in our analysis. Further details can be found in~\cite{karatzas1991brownian,borodin2002handbook}.

\renewcommand{\theequation}{A\arabic{equation}}
\setcounter{equation}{0}

\subsection{Filtrations, Stopping Times, and Martingales}\label{app:filtrations}
Let $(\Omega, \mathcal{F}, \mathbf{P})$ denote a probability space, where $\Omega$ is the sample space, $\mathcal{F}$ is a $\sigma$-field on $\Omega$, and $\mathbf{P}$ is a probability measure on $(\Omega,\mathcal{F})$. A stochastic process is a family $X = (X_t)_{t\geq 0}$ of random variables $X_t: \Omega \to \mathbb{R}$. Given a point $\omega\in\Omega$, the mapping $t \mapsto X_t(\omega)$ is called a realization or trajectory of the process $X$. The process $X$ is said to be \emph{continuous} if its trajectories are continuous almost surely, i.e.,
\begin{equation*}
    \mathbf{P}(\{ \omega\in\Omega: t \mapsto X_t(\omega)~\textrm{is continuous}\}) = 1.
\end{equation*}

A family of $\sigma$-fields $(\mathcal{F}_t)_{t\geq 0}$ is called a \emph{filtration} if
\begin{equation*}
    \mathcal{F}_s \subseteq \mathcal{F}_t \subseteq \mathcal{F},
    \quad 0\leq s\leq t.
\end{equation*}
In particular, $\mathcal{F}_t$ contains all events that can be determined using the information available up to time $t$. It is often convenient to work with a right-continuous filtration. A filtration is said to be \emph{right-continuous} if, for every $t\geq0$,
\begin{equation*}
    \mathcal{F}_t = \mathcal{F}^{+}_{t} := \bigcap_{s > t} \mathcal{F}_s.
\end{equation*}

The process $X$ is said to be \emph{adapted} to the filtration $(\mathcal{F}_t)_{t\geq0}$ if $X_t$ is $\mathcal{F}_t$-measurable for every $t\geq0$. Note that $X$ is always adapted to its \emph{natural filtration} $\mathcal{F}^0_t:= \sigma(X_s: s \leq t)$ and the right-continuous modification $\mathcal{F}_t = \bigcap_{s>t}\mathcal{F}_s^{0}$.

A random time $T:\Omega \to [0,\infty)\cup \{\infty\}$ is a \emph{stopping time} with respect to $(\mathcal{F}_t)_{t\geq0}$ if
\begin{equation}
    \{T \leq t\} \in \mathcal{F}_{t}, \quad t\geq 0.
\end{equation}
In other words, whether stopping has occurred by time $t$ is determined by the information available up to time $t$. For a continuous adapted process $X$, a typical stopping time is the first exit time from an open interval, as defined in~\eqref{eq:app-exit-time}.

An adapted process $X$ is called a \emph{martingale} with respect to a filtration $(\mathcal{F}_t)_{t\geq0}$ if $\mathbb E[|X_t|]<\infty$ and
\begin{equation}
    \mathbb{E}[X_t {\,}|{\,} \mathcal{F}_s]
    = X_s, \quad 0 \leq s \leq t. \label{eq:app-martingale}
\end{equation}
That is, its expected future value, given the past, is equal to the most recent value. 

The optional stopping theorem extends this identity from deterministic times to random stopping times.

\begin{theorem}[{Optional stopping~\cite[Ch.~1.3.C]{karatzas1991brownian}}]
\label{thm:optional-stopping}
Let $X$ be a uniformly integrable martingale with respect to $(\mathcal{F}_t)_{t\geq0}$, and let $S\leq T$ be almost surely finite stopping times. Then
\begin{equation}
    \mathbb E[X_T {\,}|{\,} \mathcal{F}_S]
    =
    X_S,
    \quad \textrm{almost surely}.
    \label{eq:app-optional-stopping}
\end{equation}
\end{theorem}

\subsection{Wiener Process and Its Properties}
\label{app:wiener-preliminaries}
A standard Wiener process $(W_t)_{t\geq0}$ is a stochastic process satisfying
\begin{enumerate}
    \item $W_0=0$ almost surely;
    \item its trajectories are continuous almost surely;
    \item for every $0\leq t_0<t_1<\cdots<t_m$, the increments
    \begin{equation*}
        W_{t_1}-W_{t_0},\,
        W_{t_2}-W_{t_1},\ldots,
        W_{t_m}-W_{t_{m-1}}
    \end{equation*}
    are mutually independent and satisfy
    \begin{equation*}
        W_{t_j} - W_{t_{j-1}}
        \sim
        \mathcal N(0,t_j - t_{j-1}),
        \quad j=1,\ldots,m.
    \end{equation*}
\end{enumerate}
We work with the right-continuous natural filtration $(\mathcal{F}_t)_{t\geq0}$.

The Wiener process has the following properties.
\begin{itemize}
    \item \emph{Symmetry.} It is symmetric in distribution, i.e.,
    \begin{equation}
    (W_t)_{t\geq0} \overset{d}{=} (-W_t)_{t\geq0}.
    \end{equation}
    \item \emph{Finite moments.} Since
    $W_t\sim\mathcal{N}(0,t)$, for every $p>0$,
    \begin{equation*}
        \mathbb{E}[|W_t|^p]
        =
        t^{p/2}\mathbb{E}[|Z|^p]
        <\infty,
        \quad t\geq0,
    \end{equation*}
    where $Z\sim\mathcal{N}(0,1)$. Consequently, for every
    finite $t$,
    \begin{equation}
        \mathbb{E}\left[\int_0^t |W_s|^p\,ds\right]
        =
        \frac{t^{1+p/2}}{1+p/2}\mathbb{E}[|Z|^p]
        <\infty.
        \label{eq:app-wiener-moments}
    \end{equation}
    \item \emph{Strong Markov property.} Let $T$ be an almost surely finite $\mathcal{F}_t$-stopping time. Then the shifted process
    \begin{equation}
        \widetilde{W}_t = W_{T+t} - W_{T}, \quad t\geq0,
    \end{equation}
    is a Wiener process independent of $\mathcal{F}_{T}$. Thus, after each stopping time, the Wiener process ``starts afresh'' from its current value.
\end{itemize}

An important consequence of the strong Markov property is the reflection principle, which relates the distribution of running extrema to the distribution of the process. Define the running maximum and minimum by
\begin{equation*}
    \overline{M}_t
    :=
    \sup_{0\leq s\leq t}W_s,
    \qquad
    \underline{m}_t
    :=
    \inf_{0\leq s\leq t} W_s.
\end{equation*}

\begin{lemma}[{Reflection principle~\cite[Ch.~2.6]{karatzas1991brownian}}]
\label{lem:reflection-principle}
For $t>0$,
\begin{equation}
    \mathbf{P}(\overline{M}_t\geq x)
    =
    2\mathbf{P}(W_t\geq x)
    =
    2Q\left(\frac{x}{\sqrt{t}}\right), \quad x \geq 0,
\end{equation}
where $Q(x)$ is the Gaussian tail function. Consequently
\begin{equation}
    \overline{M}_t \overset{d}{=} - \underline{m}_t
    \overset{d}{=}|W_t|,
    \quad
    \mathbb{E}[\overline{M}_t]
    =-\mathbb{E}[\underline{m}_t]
    =\sqrt{\frac{2t}{\pi}}.
    \label{eq:app-extrema}
\end{equation}
\end{lemma}

\subsection{First Exit Time of the Wiener Process}\label{app:first-exit-time}
Consider a Wiener process $W$ starting from zero and define its \emph{first exit time} from the interval $(-\alpha,\beta)$, $\alpha, \beta>0$, by
\begin{equation}
    \tau
    :=
    \inf\{t\geq0:W_t \notin (-\alpha,\beta)\}.
    \label{eq:app-exit-time}
\end{equation}
The exit time $\tau$ is almost surely finite, and continuity implies $W_\tau \in \{-\alpha,\beta\}$ almost surely. Applying the optional stopping theorem gives $\mathbb{E}[W_\tau]=W_0=0$, and hence
\begin{equation}
    \mathbf{P}(W_\tau=\beta)
    =
    \frac{\alpha}{\alpha+\beta},
    \quad
    \mathbf{P}(W_\tau=-\alpha)
    =
    \frac{\beta}{\alpha+\beta}.
    \label{eq:app-exit-probabilities}
\end{equation}

The distribution of $\tau$ can be characterized by the Laplace transform~\cite[p.~212]{borodin2002handbook}
\begin{equation}
    \mathcal{L}_{\tau}(s)
    :=
    \mathbb{E}[e^{-s\tau}]
    =
    \frac{
        \cosh\left(
        \frac{\alpha-\beta}{2}\sqrt{2s}
        \right)
    }{
        \cosh\left(
        \frac{\alpha+\beta}{2}\sqrt{2s}
        \right)
    },
    \quad s\geq0.
    \label{eq:app-exit-laplace}
\end{equation}
Since this expression is analytic in a neighborhood of $s=0$, all moments of $\tau$ are finite and satisfy
\begin{equation}
    \mathcal{L}_{\tau}^{(k)}(0) = (-1)^k\mathbb{E}[\tau^k].
\end{equation}
For $s>0$, the Laplace transforms of the first exit time, separated according to the exit boundary, are given by
\begin{subequations}
\label{eq:app-exit-boundary-transforms}
\begin{align}
    \mathbb{E}\left[
        e^{-s\tau}\mathbb{I}\{W_\tau=-\alpha\}
    \right]
    &=
    \frac{\sinh(\beta\sqrt{2s})}
         {\sinh((\alpha+\beta)\sqrt{2s})},
    \label{eq:app-lower-exit-transform}\\
    \mathbb{E}\left[
        e^{-s\tau}\mathbb{I}\{W_\tau=\beta\}
    \right]
    &=
    \frac{\sinh(\alpha\sqrt{2s})}
         {\sinh((\alpha+\beta)\sqrt{2s})}.
    \label{eq:app-upper-exit-transform}
\end{align}
\end{subequations}
As $s\downarrow0$, they recover the exit probabilities in~\eqref{eq:app-exit-probabilities}.

Exit-time statistics describe when the process reaches the boundary, but do not specify its location before exit. The \emph{killed transition density} describes the locations of trajectories that have not yet exited and is defined by
\begin{equation*}
    p(t,x)\,dx
    :=
    \mathbf P(W_t\in dx,\tau>t),
    \quad  t>0,\,\,-\alpha<x<\beta.
\end{equation*}
Here, ``killed'' means that a trajectory ceases to
contribute once it exits the interval. This density has the spectral representation~\cite[p.~122]{borodin2002handbook}
\begin{equation}
    p(t,x)
    =
    \frac{2}{\alpha+\beta}
    \sum_{k=1}^{\infty}
    \sin\left(
        \frac{k\pi\alpha}{\alpha+\beta}
    \right)
    \sin\left(
        \frac{k\pi(x+\alpha)}{\alpha+\beta}
    \right)
    e^{-\lambda_k t},
    \label{eq:app-killed-density-general}
\end{equation}
where
\begin{equation*}
    \lambda_k
    =
    \frac{k^2\pi^2}{2(\alpha+\beta)^2},
    \quad k\geq1.
\end{equation*}
The total mass of this density is the \emph{survival probability}
\begin{equation}
    \mathbf P(\tau>t)
    =
    \int_{-\alpha}^{\beta}p(t,x)\,dx,
    \label{eq:app-survival-probability}
\end{equation}
which is the probability that the process remains
inside the interval throughout $[0,t]$.

\subsection{Stochastic Integrals}
\label{app:stochastic-integrals}

Let $W$ be a Wiener process with respect to $(\mathcal{F}_t)_{t\geq0}$. For a continuous adapted process $X$,
the \emph{stochastic integral}
\begin{equation}
    Y_t := \int_0^t X_s\,dW_s
    \label{eq:app-stochastic-integral}
\end{equation}
accumulates Wiener increments weighted by $X_s$.

The following result gives the two properties of stochastic integrals used in our analysis.

\begin{theorem}[{see~\cite[Ch.~3.1]{borodin2002handbook}}]
\label{thm:app-ito-isometry}
Suppose that $X$ is continuous and adapted, and that
\begin{equation*}
    \mathbb{E}\left[\int_0^t X_s^2\,ds\right]<\infty 
\end{equation*}
for every finite $t$. Then $Y$ is a continuous, square-integrable martingale starting from zero. In particular,
\begin{equation}
    \mathbb{E}[Y_t]=0,
    \quad
    \mathbb{E}[Y_t^2]
    =
    \mathbb{E}\left[\int_0^t X_s^2\,ds\right].
    \label{eq:app-ito-isometry}
\end{equation}
The above identities hold at an almost surely finite stopping time $T$ that satisfies $\mathbb{E}\left[\int_0^T X_s^2\,ds\right]<\infty$.
\end{theorem}

The first identity in~\eqref{eq:app-ito-isometry} states that the stochastic integral has zero mean. The second identity, called \emph{It\^o's isometry}, expresses its second moment as an ordinary time integral. 

An important consequence of Theorem~\ref{thm:app-ito-isometry} is given below. 

\begin{lemma}\label{lem:app-stopped-stochastic-integrals}
Let $\tau$ be the first exit time
in~\eqref{eq:app-exit-time}, and let $g:[0,\infty)\to\mathbb{R}$
be continuous and satisfy
\begin{equation*}
    \mathbb{E}\left[\int_0^\tau g(s)^2\,ds\right]<\infty.
\end{equation*}
Then, for every integer $m\geq0$,
\begin{equation}
    \mathbb{E}\left[
        \int_0^\tau g(s)W_s^m\,dW_s
    \right]=0.
    \label{eq:app-general-stopped-integral}
\end{equation}
\end{lemma}
\begin{IEEEproof}
Since $|W_s|\leq K:=\max\{\alpha,\beta\}$ for $s\leq\tau$,
\begin{equation*}
    \mathbb{E}\left[
        \int_0^\tau g(s)^2W_s^{2m}\,ds
    \right]
    \leq
    K^{2m}\mathbb{E}\left[
        \int_0^\tau g(s)^2\,ds
    \right]
    <\infty.
\end{equation*}
Thus, Theorem~\ref{thm:app-ito-isometry} gives~\eqref{eq:app-general-stopped-integral}.
% For $g(s)=s^k$, the required condition follows from
% \begin{equation*}
%     \mathbb{E}\left[\int_0^\tau s^{2k}\,ds\right]
%     =
%     \frac{\mathbb{E}[\tau^{2k+1}]}{2k+1}
%     <\infty.
% \end{equation*}
% For $g(s)=e^{-rs}$, it follows from
% \begin{equation*}
%     \mathbb{E}\left[\int_0^\tau e^{-2rs}\,ds\right]
%     \leq \frac{1}{2r}<\infty.
% \end{equation*}
\end{IEEEproof}

\subsection{It\^o's Formula and Polynomial Martingales}\label{app:wiener-martingales}
Let $W$ be a Wiener process. For a twice continuously differentiable function $f$, \emph{It\^o's formula}~\cite[Ch.~3.2]{borodin2002handbook} gives
\begin{equation*}
    f(W_t)
    =
    f(W_0)
    +
    \int_0^t f'(W_s)\,dW_s
    +
    \frac{1}{2}\int_0^t f''(W_s)\,ds.
\end{equation*}

For $f(x)=x^m$, with integer $m\geq2$, this becomes
\begin{equation}
    W_t^m
    -
    \frac{m(m-1)}{2}
    \int_0^t W_s^{m-2}\,ds
    =
    m\int_0^t W_s^{m-1}\,dW_s.
    \label{eq:app-power-ito}
\end{equation}
By~\eqref{eq:app-wiener-moments} with $p=2m-2$,
\begin{equation*}
    \mathbb{E}\left[
        \int_0^t |W_s|^{2m-2}\,ds
    \right]<\infty.
\end{equation*}
By Theorem~\ref{thm:app-ito-isometry}, the stochastic integral on the RHS of~\eqref{eq:app-power-ito} is a square-integrable martingale. Consequently,
\begin{equation*}
    W_t^m
    -
    \frac{m(m-1)}{2}
    \int_0^t W_s^{m-2}\,ds
\end{equation*}
is a martingale. This result provides a convenient way to evaluate time-integrated moments of the Wiener process up to a first exit time. Consider the first exit time in~\eqref{eq:app-exit-time}. Applying the optional stopping theorem to the above martingale gives
\begin{equation}
    \mathbb{E}[W_{\tau}^m] =
    \frac{m(m-1)}{2} \mathbb{E}\left[
        \int_0^{\tau} W_s^{m-2}\,ds
    \right]. \label{eq:app-poly-martingale}
\end{equation}

\bibliographystyle{IEEEtran}
\bibliography{ref}

\end{document}